\documentclass[journal]{IEEEtran}
\IEEEoverridecommandlockouts
\usepackage{cite}
\usepackage{amsmath,amssymb,amsfonts}
\usepackage{amsthm}

\usepackage{algorithm}
\usepackage{algorithmic}
\usepackage{graphicx}
\usepackage{textcomp}
\usepackage{xcolor}
\usepackage{comment}
\newtheorem{assumption}{Assumption}
\newtheorem{lemma}{Lemma}
\newtheorem{remark}{Remark}
\usepackage{amsthm}
\newtheorem{theorem}{Theorem}
\def\BibTeX{{\rm B\kern-.05em{\sc i\kern-.025em b}\kern-.08em
    T\kern-.1667em\lower.7ex\hbox{E}\kern-.125emX}}
\begin{document}

\title{Data-Driven Time-Varying Control Barrier Functions for Adaptive Safe-Set Learning with Online Decremental Support Vector Machines\\
}
\author{
Shawon Dey, Michael Budihartono, and Hever Moncayo%
\thanks{

Shawon Dey, Michael Budihartono, and  Hever Moncayo are with the Department of Aerospace Engineering, Embry-Riddle Aeronautical University, Daytona Beach, FL 32114 USA. 
(e-mail: shawon.dey.apee@gmail.com, budiharm@my.erau.edu, moncayoh@erau.edu).}
\thanks{This work is currently being supported by the NASA University Leadership Initiative under grant No. 80NSSC25M7104.}
}
\maketitle

\begin{abstract}
Mission-critical intelligent systems often operate under time-varying limitations that reduce control authority and change the admissible safe operating envelope. In such settings, a safety certificate learned under nominal conditions may become invalid as system capability changes. To address this challenge, this paper proposes a degradation-aware, data-driven safety-filtering framework that learns a safe set from data, updates it online, and enforces the resulting learned barrier through a time-varying control barrier function (CBF). A nominal safe envelope is first learned from operational data using a radial basis function (RBF)-kernel support vector machine (SVM), whose decision function serves as the initial CBF candidate. To capture capability-induced safe-set contraction, a continuous-time decremental SVM update law is developed so that selected support-vector coefficients are reduced according to a degradation signal. A homotopy-smoothed SVM-CBF is then introduced to avoid discontinuous changes in the learned barrier during active-set transitions. The resulting time-varying learned barrier is enforced using a quadratic-program-based safety filter under degraded input constraints. Forward invariance of the learned time-varying safe set and recursive feasibility of the safety filter are established. Simulation results on a vertical takeoff and landing (VTOL) model show that the proposed method maintains safety under reduced control authority and avoids abrupt barrier-switching effects during safe-set contraction.
\end{abstract}
\section{Introduction}
\begin{figure}[htb]
    \begin{center}
        \includegraphics[width=1\linewidth]{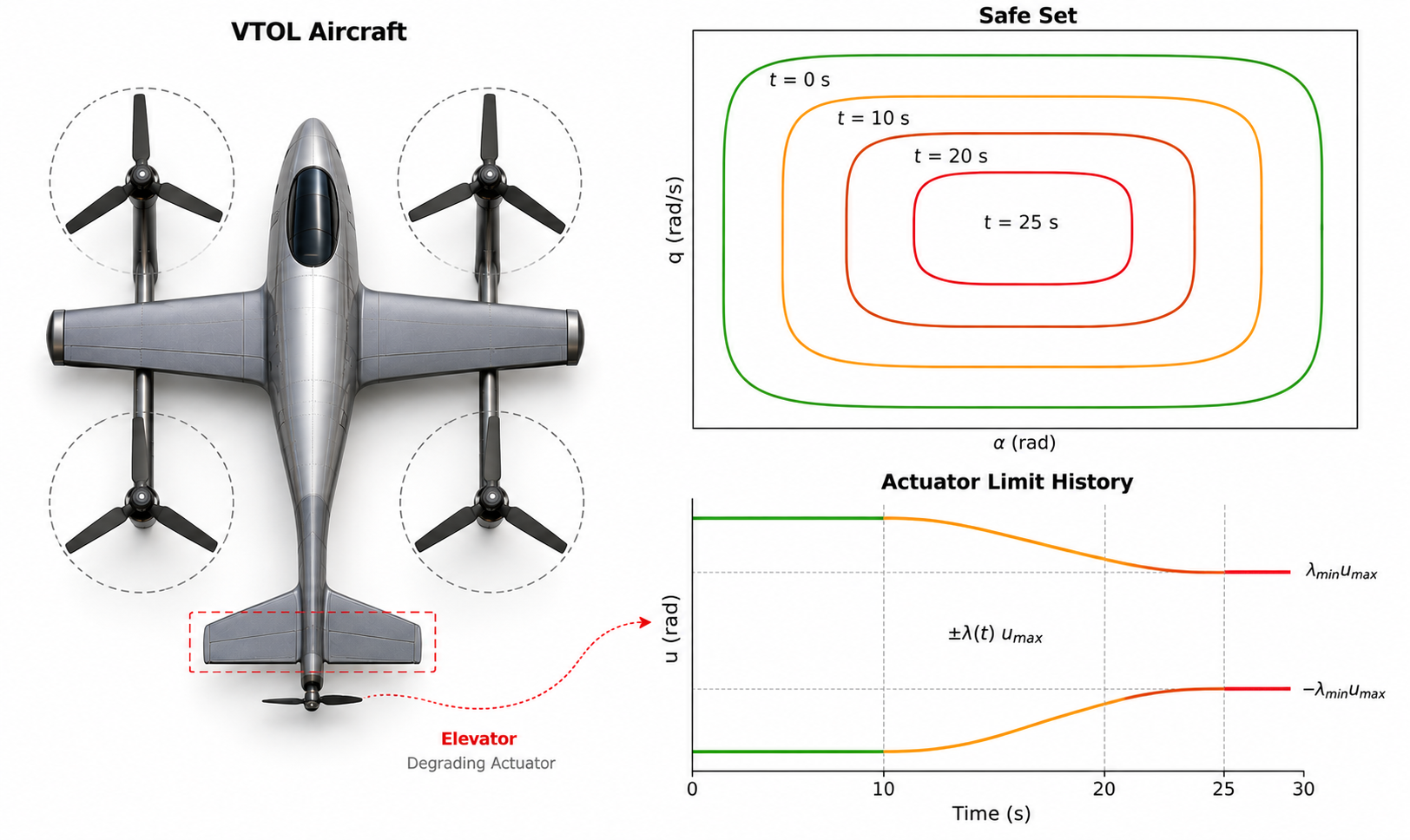}
    \end{center}
    \caption{Safe-set contraction under elevator actuator degradation. As the available elevator authority decreases, the admissible safe envelope in the pitch-state space contracts accordingly.}
    \label{fig:intro1}
\end{figure}
Safety-critical systems, such as advanced aerial vehicles \cite{namuduri2022advanced, paul2023formal}, multirotor \cite{zhou2021safety}, and robotic systems \cite{xiao2023barriernet} operating in harsh environments, require safety guarantees that remain reliable when system capability and operating conditions evolve. This need is becoming increasingly important with the rapid development of autonomy across air, ground, and space domains for applications such as urban transportation, emergency response, infrastructure inspection, delivery, surveillance, exploration, and defense \cite{wei2024autonomous}. In these domains, safety requires the vehicle to remain within its admissible operational envelope throughout the mission. Exceeding this envelope can degrade operational performance, and potentially lead to unsafe or unrecoverable behavior. In real-world scenarios, however, actuator effectiveness or sensor performance might not always remain nominal. For example, actuator or sensor degradation, external disturbances, environmental changes, and moving obstacles could reduce the effective control authority or alter the admissible operating conditions. Consequently, the safe operating envelope evolves rather than remaining fixed, as shown in Fig.~\ref{fig:intro1}. As a result, a state regime that is safe under nominal conditions becomes unsafe under degraded or uncertain operating conditions, making static safety certificates increasingly inadequate for safety-critical autonomous operation \cite{zhang2021safe, dey2025online, sun2024safety}. 

To address these challenges, control barrier functions (CBFs) \cite{Ames2017CBFQP,ames2019control} have recently received attention as an effective tool for enforcing safety. By enforcing forward invariance of a prescribed safe set \cite{nagumo1942lage}, CBFs enable real-time safety filters that minimally modify nominal control inputs while maintaining safety \cite{wang2025safe, dacs2025robust}. Nevertheless, most existing CBF-based safety rely on barrier functions designed offline from known models, predetermined fixed constraints, or analytically specified safe sets \cite{molnar2025collision}. These constructions are often insufficient when the safe operating envelope is unknown in closed form or changes during operation due to dynamically-changing operational environments. In such cases, the barrier function must be updated online using operational data, motivating data-driven time-varying CBF construction. However, incorporating new data can introduce abrupt changes in the learned decision boundary, which can weaken the regularity required for forward-invariance guarantees and real-time CBF implementation. 

Recent studies have expanded the CBF framework in several directions to address uncertainty, learning, and time-varying safety requirements. Adaptive CBFs have been developed to address parametric uncertainty and time-varying control bounds through modified barrier conditions \cite{Xiao2021AdaCBF}. However, these methods typically require a prescribed barrier function and do not address how the safe operating envelope is learned from data and updated as system capability changes. In parallel, learning-based CBF methods have been introduced to reduce the dependence on analytically designed barrier functions. For example, supervised learning and expert-demonstration-based methods synthesize barrier functions from safe trajectories and data-dependent certificate conditions \cite{srinivasan2020synthesis, robey2020learning}. Besides, more recent neural CBF approaches further improve scalability by learning policy-value-function-based safety filters for high-dimensional and input-constrained robotic systems \cite{so2024train}.

Nevertheless, the learning-based CBF methods focus on learning a fixed barrier offline and do not address how the learned safe set evolves online under changing, degraded conditions. Time-varying and compositional barrier methods have begun to address online safe-set variation, with Safari and Hoagg considering smooth switching among perception-based local CBFs in unmapped and dynamic environments \cite{safari2025time}. This work highlights the importance of smooth barrier updates for preserving the regularity needed for forward-invariance guarantees. However, the present setting requires more than smooth barrier switching, since the learned safe envelope must evolve consistently with the time-varying admissible input set. As control authority decreases, the safe set must contract so that safety remains enforceable under the remaining inputs. 

This motivates a novel degradation-scheduled decremental support vector machine (SVM) framework that adapts the learned safe envelope online as a data-driven representation of safety boundaries, while preserving the regularity required for a time-varying CBF-QP safety filter. In our preliminary work, we developed an Health Monitoring System (HMS)-informed CBF safety filter that learns a nominal flight envelope from operational data and interprets the resulting Radial Basis Function (RBF)-SVM decision boundary as a CBF candidate \cite{budihartono2026health}. However, the resulting learned envelope is still essentially static once training is complete. In this paper, we propose an online SVM update mechanism that can modify the learned boundary without retraining the classifier from scratch. Incremental and decremental SVM learning provides such a mechanism by adding or removing selected data points while preserving the Karush--Kuhn--Tucker (KKT) conditions required for maintaining the optimality of the SVM solution as the training set is updated online. \cite{cauwenberghs2000incremental}. For safety control under reduced authority, decremental updates reduce selected support-vector influence so that the nominal learned envelope contracts with the available control authority. Despite this potential, directly applying incremental or decremental SVM updates in a safety-critical control loop is not sufficient. Standard SVM maintenance algorithms update support-vector coefficients through finite steps and may change the active set abruptly. As a result, the learned decision boundary can exhibit discontinuous changes over time. Such behavior is undesirable for CBF-based safety filters because the time-varying CBF condition requires a well-defined barrier value, spatial gradient, and time derivative. Therefore, online SVM adaptation for safety-critical control requires a continuous-time boundary-evolution law that preserves the regularity needed for forward-invariance analysis and real-time Quadratic Programming (QP) implementation.

To address this gap, this paper develops a degradation-aware time-varying SVM-CBF framework. The proposed method replaces discrete decremental SVM updates with a continuous degradation-driven law, allowing selected SVM coefficients to evolve with a degradation signal that captures safe-envelope contraction. On fixed-active-set intervals, SVM KKT sensitivities update the bias and margin coefficients, giving an explicit time derivative of the learned decision function. To regularize transitions between consecutive active-set representations, a homotopy-smoothed SVM-CBF is introduced to ensure continuous barrier evolution across active-set switches. The resulting time-varying SVM-CBF-QP safety filter therefore accounts for both the time-varying learned safe set and the degraded polytopic input constraint, providing a unified mechanism for enforcing safety under changing actuation capability. The major contributions are summarized as follows.
\begin{itemize}
    \item A data-driven nominal safe envelope is learned using an RBF-kernel SVM, whose decision function serves as the initial barrier candidate.
    \item An online continuous-time decremental SVM update law is developed to adapt support-vector coefficients under degradation while preserving differentiability of the learned safe boundary.
    \item A homotopy-smoothed time-varying SVM-CBF is formulated to regularize active-set transitions and preserve the differentiability of the learned barrier during online boundary updates.
    \item A degradation-aware time-varying SVM-CBF-QP safety filter is designed to minimally modify the nominal input while enforcing the homotopy-smoothed learned barrier under input constraints.
    \item Forward invariance and recursive feasibility are established for the proposed safety filter.
    \item Simulation studies on a VTOL model demonstrate learned safe-envelope evolution and safe operation under reduced control authority.
\end{itemize}
\section{Background and Preliminaries}

\textit{Control Barrier Function:}
Consider a control-affine nonlinear system of the form,
\begin{equation} \label{dyna}
    \dot{x} = f(x) + g(x)u,
\end{equation}
where \(x \in \mathcal{X} \subset \mathbb{R}^n\) denotes the system state and \(u \in U\) denotes the control input. In the CBF framework \cite{ames2019control}, safety is characterized through a forward-invariant safe set defined by a continuously differentiable function \(h:\mathcal{X}\to\mathbb{R}\). The corresponding safe set is defined as $\mathcal{S} = \{x \in \mathcal{X} : h(x)\geq 0\}$, with boundary $\partial \mathcal{S} = \{x \in \mathcal{X} : h(x)=0\}$ and interior $\mathrm{Int}(\mathcal{S}) = \{x \in \mathcal{X} : h(x)>0\}$. Here, $\partial \mathcal{S}$ denotes the boundary of $\mathcal{S}$, while $\mathrm{Int}(\mathcal{S})$ denotes its interior. A continuously differentiable function $h$ is called a control barrier function if there exists an extended class-$\mathcal{K}$ function $\alpha$ such that
\begin{align}
\label{barriercond}
    \sup_{u \in U} \left[ L_f h(x) + L_g h(x)u \right] \geq -\alpha(h(x)),
\end{align}
for all $x \in \mathcal{X}$, where
$L_f h(x) = \frac{\partial h}{\partial x} f(x)$ and
$L_g h(x) = \frac{\partial h}{\partial x} g(x)$ are the Lie derivatives of $h$ along $f$ and $g$, respectively.

\noindent\textbf{\textit{Definition 1:}} A function $\alpha:\mathbb{R}\to\mathbb{R}$ is said to be an extended class-$\mathcal{K}$ function if it is strictly increasing and satisfies $\alpha(0)=0$ \cite{ames2019control}.

The set of control inputs that satisfy the CBF condition is
\begin{align}
\label{kcbf}
    K_{\mathrm{cbf}}(x) = \{u \in U : L_f h(x) + L_g h(x)u + \alpha(h(x)) \geq 0\}.
\end{align}
Any input $u \in K_{\mathrm{cbf}}(x)$ ensures that the safe set remains forward invariant. This result is established using Nagumo's theorem \cite{nagumo1942lage}, together with the regularity condition $\frac{\partial h}{\partial x}(x) \neq 0$ for all $x \in \partial \mathcal{S}$, as discussed in \cite{ames2019control}. In particular, Nagumo's theorem \cite{nagumo1942lage} states that forward invariance is guaranteed when the system vector field is directed inward or tangent to the boundary of the safe set.
\section{Problem Formulation}
For the continuous-time control-affine nonlinear system defined in \eqref{dyna}, the drift dynamics $f:\mathbb{R}^n \to \mathbb{R}^n$ and the control effectiveness matrix $g:\mathbb{R}^n \to \mathbb{R}^{n\times m}$ are assumed to be locally Lipschitz continuous \cite{vamvoudakis2010online}. Also, let \(u_{\mathrm{nom}}(x,t)\) denote the nominal control input. Under nominal operating conditions, the control input is constrained to the compact admissible input set
\begin{equation}
\mathcal{U}_0
=
\left\{
u\in\mathbb{R}^m:
A_u u \leq b_u
\right\},
\end{equation}
where $A_u \in \mathbb{R}^{p\times m}$ and $b_u \in \mathbb{R}^{p}$ define the nominal compact polytopic input set. Let $\mathcal{S}_0 \subset \mathbb{R}^n$ denote the nominal safe set associated with the input constraint $u \in \mathcal{U}_0$. The nominal safe set $\mathcal{S}_0$ is not known in closed form and is instead learned from nominal operational data using an SVM decision function. To model the loss of control effectiveness, we introduce a time-varying degradation parameter $\lambda(t) \in [\lambda_{\min},1]$, where $\lambda_{\min}>0$. In this work, $\lambda(t)$ is used specifically to characterize actuator degradation. However, the proposed framework can be extended to other sources of performance degradation, such as sensor degradation or other time-varying system limitations, provided that their effect on the admissible operating conditions can be appropriately parameterized. The parameter $\lambda(t)$ denotes the remaining actuator effectiveness, with $\lambda(t)=1$ corresponding to nominal authority and smaller values indicating reduced control authority. The degradation rate is defined as
\begin{equation}
    v_d(t) = -\dot{\lambda}(t) \geq 0.
\end{equation}
Accordingly, the degraded admissible input set is modeled as
\begin{equation}
\mathcal{U}(t)
=
\left\{
u\in\mathbb{R}^m:
A_u u \leq \lambda(t)b_u
\right\}.
\end{equation}
Here, \(A_u u \leq b_u\) represents the nominal input limits, and \(\lambda(t)\) contracts the admissible control set as actuator authority decreases. A scheduling signal \(\lambda_s(t)\) is introduced to regulate the temporal evolution of the learned SVM boundary. Its dynamics are defined by
\begin{equation}
    \dot{\lambda}_s(t)=-v_s(t),
    \label{eq:lambdas_dynamics}
\end{equation}
where \(v_s(t)\) is a bounded scheduling rate selected as
\begin{equation}
    v_s(t)
    =
    (1+\varepsilon_s)v_d(t),
    \qquad
    0<\varepsilon_s\ll 1,
    \label{eq:vs_buffered_rate}
\end{equation}
subject to the rate bound \(0 \leq v_s(t) \leq \bar v_s\), where \(\bar v_s>0\).
With the initialization \(\lambda_s(t_0)=\lambda(t_0)\), this construction gives \(\lambda_s(t)\leq \lambda(t)\) for all \(t\geq t_0\). Thus, \(\lambda(t)\) determines the instantaneous admissible input set \(\mathcal U(t)\), while \(\lambda_s(t)\) governs a bounded contraction of the learned SVM boundary. As actuator degradation reduces control authority, the set of states from which safety can be maintained evolves over time. Hence, the nominal safe set \(\mathcal{S}_0\) is replaced by a degradation-aware time-varying safe set \(\hat {\mathcal{S}}(t)\). This set is defined by an online-updated SVM boundary governed by \(\lambda_s(t)\). The learned boundary is represented by \(h:\mathbb{R}^n \times \mathbb{R}_{\geq 0} \to \mathbb{R}\), whose estimated set defines the degradation-aware safe region,
\begin{equation}
    \widehat {\mathcal{S}}(t)
    =
    \left\{
    x\in\mathbb{R}^n : h(x,t)\geq 0
    \right\}.
    \label{eq:time_varying_safe_set}
\end{equation}
The time dependence of \(h(x,t)\) represents the contraction of the learned safe envelope as \(\lambda_s(t)\) decreases. This contraction is implemented through the continuous-time online decremental SVM update by reducing selected support-vector influence. Since SVM active-set changes can alter the local boundary parameterization, the learned barrier is active-set dependent and is not globally continuously differentiable in time over the full horizon. Therefore, the operating horizon is partitioned according to the switching instants at which the SVM active set changes. Let \(t_0=\tau_0 < \tau_1 < \cdots < \tau_L < \tau_{L+1}=T\) denote the ordered sequence of active-set switching times over the finite horizon \([t_0,T]\). These switching times induce the intervals \(I_\ell=[\tau_\ell,\tau_{\ell+1})\), \(\ell=0,1,\ldots,L\).
On each interval \(I_\ell\), the active set is fixed and the learned decision function is denoted by \(h_\ell(x,t)\). The corresponding estimated local safe set is
\begin{equation}
    \widehat {\mathcal{S}}_\ell(t)
    =
    \left\{
    x\in\mathbb{R}^n : h_\ell(x,t)\geq 0
    \right\},
    \qquad t\in I_\ell.
    \label{eq:local_safe_set}
\end{equation}
Thus, \(\widehat{\mathcal{S}}(t)=\widehat{\mathcal{S}}_\ell(t)\) on \(I_\ell\). On each fixed-active-set interval, the boundary evolves according to parameter variation. At active-set switching times, however, the representation of the learned boundary may change. Therefore, forward invariance cannot be established solely from a standard smooth time-varying CBF condition applied on each interval, the effect of active-set transitions is also considered in this study. Given \(x(t_0)\in \widehat{\mathcal{S}}(t_0)\), the control objective is to find \(u(t)\in\mathcal{U}(t)\) such that
\begin{equation}
    x(t)\in \widehat{\mathcal{S}}(t),
    \qquad \forall t \in [t_0,T].
    \label{eq:safety_objective}
\end{equation}
This is achieved by constructing a degradation-aware online SVM-CBF \(h(x,t)\), with \(\widehat{\mathcal{S}}(t)=\{x\in\mathbb{R}^n:h(x,t)\geq 0\}\), and enforcing it through a time-varying SVM-CBF-QP safety filter under the degraded input constraint \(\mathcal{U}(t)\). Also, active-set transitions in the online decremental SVM update are addressed using a novel homotopy-smoothed barrier construction to preserve regularity.
\section{Data-Driven Time-Varying Safe Set Construction}
\label{sec:tv_safe_set_construction}
This section develops a data-driven degradation-aware safe operating envelope. A nominal safe set is first learned from operational data using an SVM, whose decision function provides an implicit representation of the nominal safe boundary. This learned boundary then serves as the baseline for the online degradation-aware update mechanism developed below.
\subsection{SVM-Based Nominal Safe Set Learning}
\label{subsec:svm_nominal_safe_set}
Under nominal operating conditions, let
\[
\mathcal{D}_0 = \{(x_i,y_i)\}_{i=1}^N, 
\qquad 
x_i \in \mathbb{R}^n, 
\quad 
y_i \in \{-1,+1\},
\]
denote a labeled dataset collected from the system, with \(y_i=+1\) for safe samples and \(y_i=-1\) for unsafe samples. The objective is to learn an SVM boundary separating safe and unsafe samples, thereby defining the nominal safe operating envelope from data. To this end, we employ a soft-margin SVM with the radial basis function (RBF) kernel
\begin{equation}
    K(x,x_i) = \exp\!\left(-\gamma \|x-x_i\|^2\right),
    \label{eq:rbf_kernel}
\end{equation}
where \(\gamma>0\) is the kernel width parameter. The resulting nominal SVM decision function is given by
\begin{equation}
    h_0(x) 
    =
    \sum_{i \in \mathcal{V}_0} \alpha_i y_i K(x,x_i) + b,
    \label{eq:nominal_decision_function}
\end{equation}
where \(\mathcal{V}_0=\{i:\alpha_i>0\}\) denotes the nominal support-vector index set, \(\alpha_i\geq 0\) are the dual coefficients, and \(b\in\mathbb{R}\) is the bias term. The learned nominal safe set is then defined as
\begin{equation}
    \widehat{\mathcal{S}}_0
    =
    \left\{
    x \in \mathbb{R}^n : h_0(x) \geq 0
    \right\}.
    \label{eq:nominal_safe_set_svm}
\end{equation}
Thus, \(h_0(x)\) provides an implicit data-driven approximation of the nominal safe operating region over the domain $\mathcal{X}$. Besides, the nominal safe set is nonempty.

\begin{assumption}
\label{ass:nominal_dataset}
The nominal dataset \(\mathcal{D}_0\) is sufficiently representative over the domain of interest, and the learned SVM decision function \(h_0(x)\) provides an approximation of the nominal safe operating envelope for the state $x$.
\end{assumption}

The SVM dual coefficients are obtained by solving the convex quadratic program
\begin{equation}
\begin{aligned}
    \min_{\alpha} \quad
    & \frac{1}{2}\sum_{i=1}^{N}\sum_{j=1}^{N}\alpha_i Q_{ij}\alpha_j
    -
    \sum_{i=1}^{N}\alpha_i \\
    \text{s.t.} \quad
    & 0 \leq \alpha_i \leq C, 
    \qquad i=1,\dots,N, \\
    & \sum_{i=1}^{N} y_i \alpha_i = 0,
\end{aligned}
\label{eq:svm_dual_problem}
\end{equation}
where \(Q_{ij}=y_i y_j K(x_i,x_j)\) and \(C>0\) is the soft-margin penalty parameter. To derive the optimality conditions used in the decremental update, define the Lagrangian for \eqref{eq:svm_dual_problem} as
\begin{equation}
    W(\alpha,b)
    =
    \frac{1}{2}\sum_{i=1}^{N}\sum_{j=1}^{N}\alpha_i Q_{ij}\alpha_j
    -
    \sum_{i=1}^{N}\alpha_i
    +
    b\sum_{i=1}^{N} y_i \alpha_i .
    \label{eq:svm_lagrangian}
\end{equation}
Here, \(b\) is the Lagrange multiplier associated with the equality constraint and corresponds to the SVM bias term. The stationarity condition with respect to \(\alpha_i\) gives
\begin{equation}
    g_i 
    =
    \frac{\partial W}{\partial \alpha_i}
    =
    \sum_{j=1}^{N} Q_{ij}\alpha_j
    +
    y_i b
    -
    1
    =
    y_i h_0(x_i)-1,
    \label{eq:kkt_gi}
\end{equation}
while differentiating with respect to \(b\) yields
\begin{equation}
    \frac{\partial W}{\partial b}
    =
    \sum_{i=1}^{N} y_i \alpha_i
    =
    0.
    \label{eq:kkt_bias_constraint}
\end{equation}
The KKT residuals \(g_i\), together with the box constraints on \(\alpha_i\), induce the standard partition of the training samples into margin, error, and reserve sets:
\begin{align}
    \mathcal{M}_0
    &=
    \left\{
    i: 0<\alpha_i<C,\; g_i=0
    \right\},
    \label{eq:nominal_margin_set}
    \\
    \mathcal{E}_0
    &=
    \left\{
    i: \alpha_i=C,\; g_i\leq 0
    \right\},
    \label{eq:nominal_error_set}
    \\
    \mathcal{R}_0
    &=
    \left\{
    i: \alpha_i=0,\; g_i\geq 0
    \right\}.
    \label{eq:nominal_reserve_set}
\end{align}
The sets \(\mathcal{M}_0\), \(\mathcal{E}_0\), and \(\mathcal{R}_0\) contain margin samples, upper-bound samples, and zero-coefficient samples, respectively. These sets define the active-set structure used in the decremental update mechanism. Under actuator degradation, the nominal learned safe set \(\widehat{\mathcal{S}}_0\) may no longer be compatible with the contracted admissible control set. Therefore, the SVM boundary must be updated as degradation evolves. Since this boundary serves as the basis for the degradation-aware safe-set construction, the following lemma establishes the local Lipschitz continuity of the nominal SVM decision function.
\begin{lemma}
\label{lem:h0_regular}
Let \(h_0:\mathbb{R}^n \to \mathbb{R}\) be the nominal SVM decision function defined in \eqref{eq:nominal_decision_function}, with kernel \(K\) given by \eqref{eq:rbf_kernel}. Then \(h_0\) is infinitely differentiable with respect to \(x\). Moreover, for every compact set \(\mathcal{X}_c \subset \mathbb{R}^n\), there exists \(L_{h_0}>0\) such that \(|h_0(x_1)-h_0(x_2)| \leq L_{h_0}\|x_1-x_2\|\) for all \(x_1,x_2 \in \mathcal{X}_c\). Consequently, \(h_0\) is locally Lipschitz on \(\mathbb{R}^n\), and its Lie derivatives along the system vector fields are well defined on the domain of interest.
\end{lemma}

\begin{proof}
Since the RBF kernel $K(x,x_i)=\exp(-\gamma \|x-x_i\|^2)$ is \(C^\infty\) in \(x\) \cite{buhmann2000radial}, and \(h_0(x)\) is a finite linear combination of such kernels plus a constant bias term, it follows that \(h_0(\cdot)\) is \(C^\infty\). Next, let \(\mathcal{X}_c \subset \mathbb{R}^n\) be compact. The gradient of the RBF kernel with respect to \(x\) is
\begin{equation}
    \nabla_x K(x,x_i)
    =
    -2\gamma (x-x_i)
    \exp(-\gamma \|x-x_i\|^2),
    \label{eq:grad_rbf}
\end{equation}
which is continuous in \(x\). Since \(\mathcal{X}_c\) is compact, \(\|\nabla_x K(x,x_i)\|\) is bounded on \(\mathcal{X}_c\) for each \(i\). Therefore,
\begin{equation}
    \nabla_x h_0(x)
    =
    \sum_{i \in \mathcal{V}_0}
    \alpha_i y_i \nabla_x K(x,x_i)
    \label{eq:grad_h0}
\end{equation}
is also bounded on \(\mathcal{X}_c\). Define \(L_{h_0}=\sup_{x \in \mathcal{X}_c}\|\nabla_x h_0(x)\|<\infty\). Then, by the mean value theorem \cite{apostol1958mathematical}, $|h_0(x_1)-h_0(x_2)|
\leq
L_{h_0}\|x_1-x_2\|,
\forall x_1,x_2 \in \mathcal{X}_c$. Hence, \(h_0\) is Lipschitz on \(\mathcal{X}_c\). Since \(\mathcal{X}_c\) is arbitrary, \(h_0\) is locally Lipschitz on \(\mathbb{R}^n\).
\end{proof}

\subsection{Decremental SVM Update Under Actuator Degradation}
\label{subsec:decremental_svm_update}
The nominal SVM decision function in \eqref{eq:nominal_decision_function} represents the safe envelope under full actuator authority. As degradation reduces the admissible input set, this boundary must be updated to remain compatible with the available control authority. The proposed update is motivated by incremental and decremental SVM learning~\cite{cauwenberghs2000incremental}, where selected support-vector influence is adjusted while preserving the SVM KKT conditions. In the decremental setting, reducing the coefficient of a selected support vector provides a mechanism for contracting the learned safe set as actuator effectiveness decreases. Next, consider a decremental update step, and let \(c\) denote the index of the selected support vector whose coefficient \(\alpha_c\) is to be reduced. During this step, \(c\) is treated as the decremental parameter and is excluded from the active margin set \(\mathcal{M}\). The set \(\mathcal{M}\) is held fixed, while the bias \(b\) and the coefficients \(\alpha_j\), \(j\in\mathcal{M}\), are adjusted to preserve the active KKT conditions. The remaining non-margin coefficients that are nonzero are held fixed during the current decremental step, while zero-coefficient reserve samples do not contribute to the SVM decision function. Since margin samples satisfy \(g_i=0\) for all \(i\in\mathcal{M}\), the decremental update imposes
\begin{equation}
    \Delta g_i = 0,
    \qquad
    \forall i\in\mathcal{M}.
    \label{eq:delta_gi_zero_margin}
\end{equation}
Using the KKT residual definition in \eqref{eq:kkt_gi}, the induced variation in \(g_i\) due to changes in \(\alpha_c\), \(\alpha_j\), \(j\in\mathcal{M}\), and \(b\) is
\begin{equation}
    \Delta g_i
    =
    Q_{ic}\Delta \alpha_c
    +
    \sum_{j\in\mathcal{M}} Q_{ij}\Delta \alpha_j
    +
    y_i\Delta b .
    \label{eq:delta_gi}
\end{equation}
In addition, the equality constraint \eqref{eq:kkt_bias_constraint} must remain satisfied, which gives
\begin{equation}
    y_c\Delta\alpha_c
    +
    \sum_{j\in\mathcal{M}} y_j\Delta\alpha_j
    =
    0.
    \label{eq:delta_equality_constraint}
\end{equation}
Let $\mathcal{M}=\{m_1,\ldots,m_{|\mathcal{M}|}\}$ denote the current margin set. Collect the unknown parameter variations into \(\Delta\vartheta=\begin{bmatrix}\Delta b & \Delta\alpha_{m_1} & \cdots & \Delta\alpha_{m_{|\mathcal{M}|}}\end{bmatrix}^{\top}\). Then \eqref{eq:delta_equality_constraint} and \eqref{eq:delta_gi_zero_margin} can be written in the compact form
\begin{equation}
    \mathcal{Q}\Delta\vartheta
    =
    -q_c\Delta\alpha_c,
    \label{eq:augmented_decremental_system_compact}
\end{equation}
where
\begin{equation}
    \mathcal{Q}
    =
    \begin{bmatrix}
        0 & y_{m_1} & \cdots & y_{m_{|\mathcal{M}|}} \\
        y_{m_1} & Q_{m_1m_1} & \cdots & Q_{m_1m_{|\mathcal{M}|}} \\
        \vdots & \vdots & \ddots & \vdots \\
        y_{m_{|\mathcal{M}|}} & Q_{m_{|\mathcal{M}|}m_1} & \cdots & Q_{m_{|\mathcal{M}|}m_{|\mathcal{M}|}}
    \end{bmatrix},
    \label{eq:Q_augmented_def}
\end{equation}
and \(q_c=\begin{bmatrix}y_c & Q_{m_1c} & \cdots & Q_{m_{|\mathcal{M}|}c}\end{bmatrix}^{\top}\). The first row of \eqref{eq:augmented_decremental_system_compact} enforces the equality constraint, while the remaining rows enforce the active margin conditions. If \(\mathcal{Q}\) is nonsingular for the current active set~\cite{cauwenberghs2000incremental}, then \(\Delta\vartheta\) is uniquely determined by \(\Delta\alpha_c\)
\begin{equation}
    \Delta\vartheta
    =
    -\mathcal{Q}^{-1}q_c\Delta\alpha_c.
    \label{eq:decremental_solution_compact}
\end{equation}
Define the sensitivity vector \(\begin{bmatrix}\beta & \beta_{m_1} & \cdots & \beta_{m_{|\mathcal{M}|}}\end{bmatrix}^{\top}=-\mathcal{Q}^{-1}q_c\). Then the corresponding updates are
\begin{equation}
    \Delta b
    =
    \beta\Delta\alpha_c,
    \qquad
    \Delta\alpha_j
    =
    \beta_j\Delta\alpha_c,
    \quad
    \forall j\in\mathcal{M}.
    \label{eq:beta_update_law}
\end{equation}
For samples that are not in the current margin set, substituting \eqref{eq:beta_update_law} into \eqref{eq:delta_gi} yields
\begin{equation}
    \Delta g_i
    =
    ( Q_{ic}
    +
    \sum_{j\in\mathcal{M}} Q_{ij}\beta_j
    +
    y_i\beta )
    \Delta\alpha_c,
    \qquad
    i\notin\mathcal{M}.
    \label{eq:delta_gi_gamma}
\end{equation}
Accordingly, define the residual sensitivity
\begin{equation}
    \gamma_i
    =
    Q_{ic}
    +
    \sum_{j\in\mathcal{M}} Q_{ij}\beta_j
    +
    y_i\beta,
    \qquad
    i\notin\mathcal{M},
    \label{eq:gamma_definition}
\end{equation}
so that \(\Delta g_i=\gamma_i\Delta\alpha_c\) for \(i\notin\mathcal{M}\). The sensitivities \(\gamma_i\) determine when non-margin samples reach KKT boundaries and trigger active-set updates. Therefore, the discrete decremental SVM update is piecewise, \(\alpha_c\) is reduced until a KKT boundary is reached, after which the margin, error, and reserve sets are recomputed. Although this preserves the KKT conditions during each step, the resulting boundary evolution can not comply the regularity required for a time-varying CBF. This motivates the continuous-time decremental formulation developed next, which ties coefficient variation to the degradation signal and ensures differentiability.
\subsection{Continuous-Time Decremental Update Law}
\label{subsec:ct_decremental_update}
The discrete decremental update preserves the SVM KKT conditions through finite coefficient changes and active-set sensitivity recomputation. For time-varying CBF construction, however, the learned boundary must evolve continuously in time. We therefore replace the discrete decremental step with a continuous parameter evolution driven by the scheduling signal \(\lambda_s(t)\). Consider a fixed-active-set interval \(\mathcal{I}_\ell=[\tau_\ell,\tau_{\ell+1})\), over which the active margin set remains fixed. Let \(\mathcal{M}_\ell=\{m_1,\ldots,m_{|\mathcal{M}_\ell|}\}\) denote this margin set, excluding the selected decremental index \(c\), whose coefficient is treated as the degradation-driven parameter. Define the active boundary-parameter vector
\begin{equation}
    \theta(t)
    =
    \begin{bmatrix}
        b(t) &
        \alpha_{m_1}(t) &
        \cdots &
        \alpha_{m_{|\mathcal{M}_\ell|}}(t)
    \end{bmatrix}^{\top}
    \in\mathbb{R}^{|\mathcal{M}_\ell|+1}.
    \label{eq:theta_def}
\end{equation}
The selected coefficient is parameterized by the degradation-scheduling signal as \(\alpha_c(t)=\varphi_c(\lambda_s(t))\), where \(\varphi_c:[\lambda_{\min},1]\to\mathbb{R}_{\geq 0}\) is continuously differentiable, satisfies \(\varphi_c(1)=\alpha_c^0\), and obeys \(\varphi_c'(\lambda_s)\geq 0\) for \(\lambda_s\in[\lambda_{\min},1]\). Since $\dot{\lambda}_s(t)\leq 0$, this choice implies that $\alpha_c(t)=\varphi_c(\lambda_s(t))$ is non-increasing in time. Thus, as the scheduled degradation level decreases, the contribution of the selected support vector is reduced. On the fixed-active-set interval $\mathcal{I}_\ell$, the equality constraint and the active margin conditions is written as
\begin{equation}
    F(\theta(t),\lambda_s(t))=0,
    \label{eq:implicit_kkt_map}
\end{equation}
where
\begin{equation}
\begin{aligned}
    &F(\theta,\lambda_s)
    =
    \\
    &\begin{bmatrix}
        \displaystyle
        \sum_{j\in\mathcal{M}_l} y_j\alpha_j
        +
        y_c\varphi_c(\lambda_s)
        \\[1.5ex]
        \displaystyle
        \sum_{j\in\mathcal{M}_l} Q_{m_1j}\alpha_j
        +
        Q_{m_1c}\varphi_c(\lambda_s)
        +
        y_{m_1}b
        -
        1
        \\
        \vdots
        \\
        \displaystyle
        \sum_{j\in\mathcal{M}_l} Q_{m_{|\mathcal{M}_l|}j}\alpha_j
        +
        Q_{m_{|\mathcal{M}_l|}c}\varphi_c(\lambda_s)
        +
        y_{m_{|\mathcal{M}_l|}}b
        -
        1
    \end{bmatrix}      
\end{aligned}
    \label{eq:F_map_expanded}
\end{equation}
The first component of \(F\) enforces the SVM equality constraint, while the remaining components enforce \(g_i=0\) for all active margin samples \(i\in\mathcal{M}_\ell\). Differentiating \eqref{eq:implicit_kkt_map} with respect to time yields
\begin{equation}
    \frac{\partial F}{\partial \theta}(\theta(t),\lambda_s(t))\dot{\theta}(t)
    +
    \frac{\partial F}{\partial \lambda_s}(\theta(t),\lambda_s(t))\dot{\lambda}_s(t)
    =
    0.
    \label{eq:implicit_diff_kkt}
\end{equation}
Define
\begin{equation}
    H(t)
    =
    \frac{\partial F}{\partial \theta}(\theta(t),\lambda_s(t)),
    \qquad
    J(t)
    =
    \frac{\partial F}{\partial \lambda_s}(\theta(t),\lambda_s(t)).
    \label{eq:H_J_defs}
\end{equation}
For the system \eqref{eq:F_map_expanded}, these matrices satisfy
\begin{equation}
    H(t)
    =
    \begin{bmatrix}
        0 & y_{m_1} & \cdots & y_{m_{|\mathcal{M}_l|}}\\
        y_{m_1} & Q_{m_1m_1} & \cdots & Q_{m_1m_{|\mathcal{M}_l|}}\\
        \vdots & \vdots & \ddots & \vdots\\
        y_{m_{|\mathcal{M}_l|}} & Q_{m_{|\mathcal{M}_l|}m_1} & \cdots & Q_{m_{|\mathcal{M}_l|}m_{|\mathcal{M}_l|}}
    \end{bmatrix},
    \label{eq:H_explicit}
\end{equation}
and \(J(t)=\varphi_c'(\lambda_s(t))
\begin{bmatrix}
y_c & Q_{m_1c} & \cdots & Q_{m_{|\mathcal{M}_\ell|}c}
\end{bmatrix}^{\top}\). Thus, \(H(t)\) coincides with the augmented active-set matrix appearing in the decremental sensitivity calculation, while \(J(t)\) captures how the active KKT equalities change with the degradation-driven coefficient \(\varphi_c(\lambda_s(t))\). Here, \(H(t)\) is nonsingular, and \eqref{eq:implicit_diff_kkt} gives
\begin{equation}
    \dot{\theta}(t)
    =
    -H(t)^{-1}J(t)\dot{\lambda}_s(t).
    \label{eq:theta_dot_general}
\end{equation}
Using \(v_s(t)=-\dot{\lambda}_s(t)\), this can equivalently be written as
\begin{equation}
    \dot{\theta}(t)
    =
    H(t)^{-1}J(t)v_s(t).
    \label{eq:theta_dot_vs}
\end{equation}
Therefore,
\begin{equation}
    \begin{bmatrix}
        \dot{b}(t)\\
        \dot{\alpha}_{\mathcal{M}}(t)
    \end{bmatrix}
    =
    -H(t)^{-1}J(t)\dot{\lambda}_s(t)
    =
    H(t)^{-1}J(t)v_s(t).
    \label{eq:boundary_parameter_dynamics}
\end{equation}
Equation \eqref{eq:boundary_parameter_dynamics} defines the continuous-time decremental update law on a fixed-active-set interval. It provides a degradation-driven evolution of the active SVM boundary parameters while preserving the active KKT equalities.

\begin{assumption}
\label{ass:fixed_active_set_interval}
On each fixed-active-set interval \(\mathcal{I}_\ell\), the margin set \(\mathcal{M}_\ell\) remains constant, \(\lambda_s(t)\in[\lambda_{\min},1]\) is continuously differentiable, and \(\varphi_c\) is continuously differentiable. The implicit system \(F(\theta,\lambda_s)=0\) admits a solution, and \(H(t)=\frac{\partial F}{\partial\theta}(\theta(t),\lambda_s(t))\) is nonsingular for all \(t\in\mathcal{I}_\ell\).
\end{assumption}

\begin{lemma}
\label{lem:theta_differentiable}
Under Assumption~\ref{ass:fixed_active_set_interval}, for each \(t^\ast\in\mathcal{I}_\ell\), there exists a neighborhood \(\mathcal{I}^\ast\subseteq\mathcal{I}_\ell\) of \(t^\ast\) and a locally unique differentiable trajectory \(\theta(t)\) satisfying \(F(\theta(t),\lambda_s(t))=0\) for all \(t\in\mathcal{I}^\ast\). Moreover, $\theta(t)$ satisfies \eqref{eq:boundary_parameter_dynamics}, $\forall t\in\mathcal{I}^\ast$.
\end{lemma}
\begin{proof}
Fix \(t^\ast\in\mathcal{I}_\ell\). By Assumption~\ref{ass:fixed_active_set_interval}, there exists \(\theta^\ast=\theta(t^\ast)\) such that \(F(\theta^\ast,\lambda_s(t^\ast))=0\), and the Jacobian \(\frac{\partial F}{\partial \theta}(\theta^\ast,\lambda_s(t^\ast))\) is nonsingular. Since \(F\) is continuously differentiable, the implicit function theorem implies that, in neighborhoods of \(\lambda_s(t^\ast)\) and \(\theta^\ast\), the relation \(F(\theta,\lambda_s)=0\) defines a locally unique continuously differentiable mapping \(\theta=\Theta(\lambda_s)\). Since \(\lambda_s(t)\) is continuously differentiable on \(\mathcal{I}_\ell\), the composition \(\theta(t)=\Theta(\lambda_s(t))\) is differentiable on a neighborhood \(\mathcal{I}^\ast\subseteq\mathcal{I}_\ell\) of \(t^\ast\). Differentiating \(F(\theta(t),\lambda_s(t))=0\) with respect to time gives \(\frac{\partial F}{\partial \theta}\dot{\theta}(t)+\frac{\partial F}{\partial \lambda_s}\dot{\lambda}_s(t)=0\). Since \(H(t)=\partial F/\partial\theta\) is nonsingular, solving for \(\dot{\theta}(t)\) gives \(\dot{\theta}(t)=-H(t)^{-1}J(t)\dot{\lambda}_s(t)\). Using \(v_s(t)=-\dot{\lambda}_s(t)\) yields \eqref{eq:boundary_parameter_dynamics}.
\end{proof}
Lemma~\ref{lem:theta_differentiable} shows that, on fixed-active-set intervals, the SVM boundary parameters evolve differentiably with \(\lambda_s(t)\), yielding a well-defined time derivative of the learned decision boundary.
\begin{remark}
The update law \eqref{eq:boundary_parameter_dynamics} is local to a fixed-active-set interval. When a KKT boundary is reached, the active set may change and a new implicit system of the form \eqref{eq:implicit_kkt_map} is initialized. Hence, the full decremental SVM evolution is piecewise smooth in time, and active-set transitions are addressed separately through the homotopy-smoothed barrier construction used in the safety-filter design.
\end{remark}
\subsection{Time-Varying Decision Boundary as a Candidate Control Barrier Function}
\label{subsec:tv_decision_boundary_cbf}
On each fixed-active-set interval \(\mathcal{I}_\ell=[\tau_\ell,\tau_{\ell+1})\), the active margin set \(\mathcal{M}_\ell=\{m_1,\ldots,m_{|\mathcal{M}_\ell|}\}\) remains fixed. The selected decremental index \(c\) is excluded from \(\mathcal{M}_\ell\) and treated as the degradation-driven parameter. Let $\mathcal{F}_\ell$ denote the set of non-margin indices with nonzero SVM coefficients whose coefficients remain fixed on $\mathcal{I}_\ell$, i.e., 
$\mathcal{F}_\ell
=
\left\{
k:\alpha_k\neq 0,\;
k\notin\mathcal{M}_\ell\cup\{c\}
\right\}$.
The zero-coefficient reserve samples are omitted since they do not contribute to the SVM decision function. The corresponding learned SVM decision function is
\begin{equation}
\begin{aligned}
    h_\ell(x,t)
    &=
    \sum_{j\in\mathcal{M}_\ell}
    \alpha_j^\ell(t)y_jK(x_j,x)
    +
    \varphi_c(\lambda_s(t))y_cK(x_c,x)  \\
    &\quad
    +
    \sum_{k\in\mathcal{F}_\ell}
    \alpha_k y_kK(x_k,x)
    +
    b_\ell(t),
    \qquad t\in \mathcal{I}_\ell.
\end{aligned}
    \label{eq:local_tv_svm_decision_function}
\end{equation}
Here, the terms correspond to the active margin contribution, the reduced selected support-vector contribution, the fixed nonzero-coefficient contribution, and the bias \(b_\ell(t)\), respectively. The associated local degradation-aware safe set is
\begin{equation}
    \widehat {\mathcal{S}}_\ell(t)
    =
    \left\{
    x\in\mathbb{R}^n:
    h_\ell(x,t)\geq 0
    \right\},
    \qquad t\in I_\ell.
    \label{eq:local_tv_safe_set}
\end{equation}
Thus, $h_\ell(x,t)$ describes the learned safe boundary only on the interval over which the active-set structure is fixed. For \(t\in\mathcal{I}_\ell\), the active parameter vector is defined as \(\theta_\ell(t)=\begin{bmatrix} b_\ell(t) & \alpha_{\mathcal{M}_\ell}^\ell(t) \end{bmatrix}^{\top}\).
By Lemma~\ref{lem:theta_differentiable}, $\theta_\ell(t)$ is differentiable on $I_\ell$ and satisfies
\begin{equation}
    \dot{\theta}_\ell(t)
    =
    H_\ell(t)^{-1}J_\ell(t)v_s(t),
    \quad t\in I_\ell,
    \label{eq:theta_ell_dot}
\end{equation}
where $H_\ell(t)$ and $J_\ell(t)$ denote the matrices for the fixed active set $\mathcal{M}_\ell$. Moreover, since \(\alpha_c(t)=\varphi_c(\lambda_s(t))\), we have
\begin{equation}
    \dot{\alpha}_c(t)
    =
    \varphi_c'(\lambda_s(t))\dot{\lambda}_s(t)
    =
    -\varphi_c'(\lambda_s(t))v_s(t).
    \label{eq:alpha_c_dot_vs}
\end{equation}
Therefore, for each fixed state $x$, the explicit time dependence of $h_\ell(x,t)$ is induced by the evolution of the active SVM coefficients, the selected decremental coefficient, and the bias. Differentiating \eqref{eq:local_tv_svm_decision_function} with respect to time gives
\begin{equation}
\begin{aligned}
    \frac{\partial h_\ell(x,t)}{\partial t}
    &=
    \sum_{j\in\mathcal{M}_\ell}
    \dot{\alpha}_j^\ell(t)y_jK(x_j,x)
    +
    \dot{\alpha}_c(t)y_c
    \\
    &K(x_c,x)
    +
    \dot{b}_\ell(t).
\end{aligned}
    \label{eq:h_ell_t_general}
\end{equation}
The coefficients associated with $\mathcal{F}_\ell$ do not contribute to \eqref{eq:h_ell_t_general} because they are held constant on $I_\ell$. Let \(\psi_\ell(t)=H_\ell(t)^{-1}J_\ell(t)\), with components indexed consistently with \(\theta_\ell(t)=\begin{bmatrix} b_\ell(t) & \alpha_{m_1}^\ell(t) & \cdots & \alpha_{m_{|\mathcal{M}_\ell|}}^\ell(t) \end{bmatrix}^{\top}\). Then \eqref{eq:theta_ell_dot} implies
\begin{equation}
    \dot{b}_\ell(t)
    =
    \big(\psi_\ell(t)\big)_0 v_s(t),
    \label{eq:b_ell_dot_component}
\end{equation}
where the subscript $0$ denotes the first component of $\psi_\ell(t)$, corresponding to the bias term $b_\ell(t)$. Also, for each $r=1,\ldots,|\mathcal{M}_\ell|$ with $m_r\in\mathcal{M}_\ell$,
\begin{equation}
    \dot{\alpha}_{m_r}^\ell(t)
    =
    \big(\psi_\ell(t)\big)_r v_s(t).
    \label{eq:alpha_ell_dot_component}
\end{equation}
Substituting \eqref{eq:alpha_c_dot_vs}, \eqref{eq:b_ell_dot_component}, and \eqref{eq:alpha_ell_dot_component} into \eqref{eq:h_ell_t_general} yields
\begin{equation}
    \frac{\partial h_\ell(x,t)}{\partial t}
    =
    \Phi_\ell(x,t)v_s(t),
    \label{eq:h_ell_t_vs_form}
\end{equation}
where
\begin{equation}
\begin{aligned}
    \Phi_\ell(x,t)
    &=
    \sum_{r=1}^{|\mathcal{M}_\ell|}
    \big(\psi_\ell(t)\big)_r
    y_{m_r}K(x_{m_r},x)  \\
    &\quad
    +
    \big(\psi_\ell(t)\big)_0
    -
    \varphi_c'(\lambda_s(t))y_cK(x_c,x).
\end{aligned}
\label{eq:phi_ell_def}
\end{equation}
The function \(\Phi_\ell(x,t)\) gives the boundary sensitivity to \(\lambda_s(t)\) on \(\mathcal{I}_\ell\), while \(v_s(t)\) scales the resulting boundary motion.

\begin{lemma}
\label{lem:h_ell_regular_tv}
Let $I_\ell=[\tau_\ell,\tau_{\ell+1})$ be a fixed-active-set interval, and let $h_\ell(x,t)$ be defined by \eqref{eq:local_tv_svm_decision_function}. Suppose the conditions of Lemma~\ref{lem:theta_differentiable} hold on $I_\ell$. If there exist constants $M_\theta>0$, $M_s>0$, and $M_c>0$ such that
\begin{equation}
\begin{aligned}
    \|H_\ell(t)^{-1}J_\ell(t)\|
    &\leq M_\theta,\\
    |\dot{\lambda}_s(t)|
    &\leq M_s,\\
    |\varphi_c'(\lambda_s(t))|
    &\leq M_c,
    \qquad \forall t\in I_\ell.
\end{aligned}
    \label{eq:bounded_HJ_lambdasdot_ell}
\end{equation}
then the following properties hold on $I_\ell$:
\begin{enumerate}
    \item for each fixed $t\in I_\ell$, the map $x\mapsto h_\ell(x,t)$ is $C^\infty$;
    \item for each fixed $x$, the partial derivative $\partial h_\ell(x,t)/\partial t$ exists and is given by \eqref{eq:h_ell_t_general}--\eqref{eq:h_ell_t_vs_form};
    \item for every compact set $\mathcal{X}_c\subset\mathbb{R}^n$ and every compact subinterval $I_c\subset I_\ell$, $h_\ell(\cdot,t)$ is Lipschitz in $x$ uniformly in $t\in I_c$;
    \item $\partial h_\ell(x,t)/\partial t$ is bounded on $\mathcal{X}_c\times I_c$ for every compact set $\mathcal{X}_c\subset\mathbb{R}^n$ and every compact subinterval $I_c\subset I_\ell$.
\end{enumerate}
\end{lemma}

\begin{proof}
Since $I_\ell$ is a fixed-active-set interval, the index sets $\mathcal{M}_\ell$ and $\mathcal{F}_\ell$ are constant on $I_\ell$. Hence, for each fixed $t\in I_\ell$, $h_\ell(x,t)$ is a finite linear combination of RBF kernels plus a scalar bias term. Because the RBF kernel is $C^\infty$ in $x$, it follows that $h_\ell(\cdot,t)$ is $C^\infty$. For the second statement, Lemma~\ref{lem:theta_differentiable} implies that the active parameter vector $\theta_\ell(t)$ is differentiable on $I_\ell$. In addition, $\alpha_c(t)=\varphi_c(\lambda_s(t))$ is differentiable because $\varphi_c$ and $\lambda_s(t)$ are continuously differentiable. Therefore, differentiating \eqref{eq:local_tv_svm_decision_function} with respect to time gives \eqref{eq:h_ell_t_general}. Substituting the parameter dynamics \eqref{eq:theta_ell_dot} and \eqref{eq:alpha_c_dot_vs} yields \eqref{eq:h_ell_t_vs_form}. For the third statement, fix a compact set $\mathcal{X}_c\subset\mathbb{R}^n$ and a compact subinterval $I_c\subset I_\ell$. The gradient of $h_\ell$ with respect to $x$ is
\begin{equation}
\begin{aligned}
    \nabla_x h_\ell(x,t)
    &=
    \sum_{j\in\mathcal{M}_\ell}
    \alpha_j^\ell(t)y_j\nabla_x K(x_j,x)
    +
    \varphi_c(\lambda_s(t))y_c \\
    &\quad
    \nabla_x K(x_c,x)
    +
    \sum_{k\in\mathcal{F}_\ell}
    \alpha_k y_k\nabla_x K(x_k,x).
\end{aligned}
    \label{eq:grad_h_ell}
\end{equation}
For each fixed index, the function $\nabla_xK$ is continuous in $x$, and hence bounded on $\mathcal{X}_c$ \cite{rudin2021principles}. Since $\theta_\ell(t)$ is differentiable, it is continuous. Therefore, its components are bounded on the compact interval $I_c$. The coefficient $\varphi_c(\lambda_s(t))$ is also continuous and bounded on $I_c$, and the fixed coefficients $\alpha_k$, $k\in\mathcal{F}_\ell$, are constant on $I_\ell$. Since $\mathcal{M}_\ell$ and $\mathcal{F}_\ell$ are finite, there exists a constant $L_{\ell,c}>0$ such that
\begin{equation}
    \sup_{(x,t)\in\mathcal{X}_c\times I_c}
    \|\nabla_x h_\ell(x,t)\|
    \leq L_{\ell,c}.
\end{equation}
By the mean value theorem,
\begin{equation}
    |h_\ell(x_1,t)-h_\ell(x_2,t)|
    \leq
    L_{\ell,c}\|x_1-x_2\|,
    \quad
    \forall x_1,x_2\in\mathcal{X}_c.
\end{equation}
Thus, $h_\ell(\cdot,t)$ is Lipschitz in $x$ uniformly in $t\in I_c$. For the fourth statement, using \eqref{eq:h_ell_t_general} and the RBF kernel bound $0<K(x_j,x)\leq 1$, we obtain
\begin{equation}
    \left|
    \frac{\partial h_\ell(x,t)}{\partial t}
    \right|
    \leq
    \sum_{j\in\mathcal{M}_\ell}
    |\dot{\alpha}_j^\ell(t)|
    +
    |\dot{\alpha}_c(t)|
    +
    |\dot{b}_\ell(t)|.
    \label{eq:h_ell_t_bound}
\end{equation}
From \eqref{eq:theta_ell_dot} and \eqref{eq:bounded_HJ_lambdasdot_ell},
\begin{equation}
    \|\dot{\theta}_\ell(t)\|
    \leq
    \|H_\ell(t)^{-1}J_\ell(t)\|\,|\dot{\lambda}_s(t)|
    \leq
    M_\theta M_s,
    \qquad t\in I_\ell.
\end{equation}
Thus, all components of $\dot{\theta}_\ell(t)$, including $\dot{b}_\ell(t)$ and $\dot{\alpha}_j^\ell(t)$, are bounded on $I_\ell$. Also,
\begin{equation}
    |\dot{\alpha}_c(t)|
    =
    |\varphi_c'(\lambda_s(t))|\,|\dot{\lambda}_s(t)|
    \leq
    M_cM_s.
\end{equation}
Substituting these bounds into \eqref{eq:h_ell_t_bound} yields a bound on $\partial h_\ell(x,t)/\partial t$ over $\mathcal{X}_c\times I_c$. This proves the fourth statement.
\end{proof}
Lemma~\ref{lem:h_ell_regular_tv} shows that \(h_\ell(x,t)\) is smooth in \(x\) and differentiable in \(t\) on each fixed-active-set interval, making it a local candidate time-varying barrier function. However, active-set switching can change the SVM boundary representation and prevent global smoothness over the full horizon. The next subsection introduces a homotopy-smoothed barrier construction that connects consecutive fixed-active-set representations and preserves the regularity needed for the SVM-CBF-QP safety filter.
\section{Time-Varying SVM-CBF-QP Safety Filter Design}
\label{sec:safety_filter_design}
\subsection{Data-Driven Homotopy-Smoothed Time-Varying SVM-CBF}
\label{subsec:hst_svm_cbf_condition}
The fixed-active-set functions \(h_\ell(x,t)\) define the learned SVM boundary only while the active set remains unchanged. Since switching directly from \(h_\ell(x,t)\) to \(h_{\ell+1}(x,t)\) can make the barrier discontinuous and its time derivative ill-defined, active-set transitions are treated using a novel homotopy-smoothed time-varying SVM barrier \(h_{\mathrm{H}}:\mathbb{R}^n\times\mathbb{R}_{\geq 0}\to\mathbb{R}\). The function \(h_{\mathrm{H}}\) coincides with the local SVM decision function away from active-set transitions and smoothly interpolates between consecutive boundary representations during switching, as illustrated in Fig.~\ref{fig:homotopy}.
\begin{figure}[htb]
    \begin{center}
        \includegraphics[width=1\linewidth]{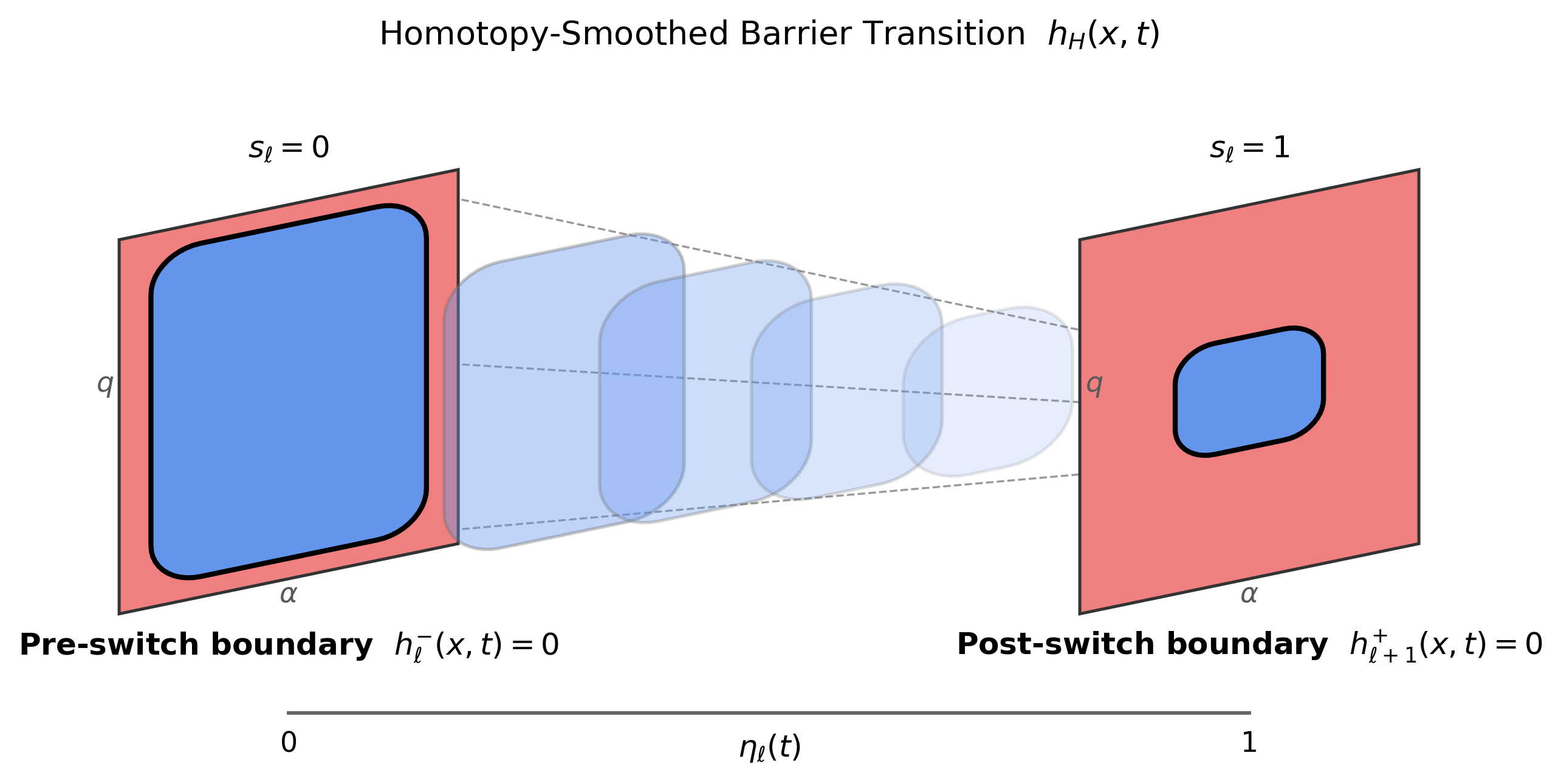}
    \end{center}
    \caption{Homotopy-smoothed barrier transition. As $\eta_\ell(t)$ increases from $0$ to $1$, the safe-set boundary transitions continuously from $h_\ell^-(x,t)$ to $h_{\ell+1}^+(x,t)$, yielding a contraction of the safe set in the $(\alpha,q)$ plane.}
    \label{fig:homotopy}
\end{figure}
The resulting safety guarantee is with respect to the homotopy-smoothed learned safe set enforced by the controller. Let \(\tau_{\ell+1}\) denote a switching time at which the local representation changes from \(h_\ell\) to \(h_{\ell+1}\). Define the transition window \(\mathcal{T}_\ell=[\tau_{\ell+1},\,\tau_{\ell+1}+T_\ell^{\mathrm{tr}}]\), where \(T_\ell^{\mathrm{tr}}>0\) is a design parameter. The transition windows are assumed to be nonoverlapping. Let \(s_\ell(t)=\frac{t-\tau_{\ell+1}}{T_\ell^{\mathrm{tr}}}\) for \(t\in\mathcal{T}_\ell\),
so that $s_\ell(t)\in[0,1]$ during the transition. An admissible homotopy function \(\eta:[0,1]\to[0,1]\) is introduced such that \(\eta\in C^1([0,1])\), \(\eta(0)=0\), \(\eta(1)=1\), and \(\eta'(0)=\eta'(1)=0\). The transition weight is then defined as \(\eta_\ell(t)=\eta(s_\ell(t))\) for \(t\in\mathcal{T}_\ell\).
If higher-order barrier constructions are required, $\eta$ can be chosen from $C^q([0,1])$ with matching derivatives up to order $q$ at the endpoints. Let $h_\ell^{-}(x,t)$ denote a continuously differentiable extension of the pre-switch boundary representation over $\mathcal{T}_\ell$, and let $h_{\ell+1}^{+}(x,t)$ denote the post-switch boundary representation initialized after the active-set update. During $\mathcal{T}_\ell$, the homotopy-smoothed SVM barrier is defined by
\begin{equation}
    h_{\mathrm{H}}(x,t)
    =
    \big(1-\eta_\ell(t)\big)h_\ell^{-}(x,t)
    +
    \eta_\ell(t)h_{\ell+1}^{+}(x,t),
    \qquad
    t\in\mathcal{T}_\ell.
    \label{eq:hH_homotopy_transition}
\end{equation}
Outside transition windows, \(h_{\mathrm{H}}\) coincides with the current fixed-active-set SVM decision function,
\begin{equation}
    h_{\mathrm{H}}(x,t)
    =
    h_\ell(x,t),
    \qquad
    t\in I_\ell\setminus \mathcal{T}_{\ell-1}.
    \label{eq:hH_fixed_interval}
\end{equation}
The corresponding homotopy-smoothed learned safe set is defined as \(\widehat{\mathcal{S}}_{\mathrm{H}}(t)=\{x\in\mathbb{R}^n:h_{\mathrm{H}}(x,t)\geq 0\}\). The time derivative of $h_{\mathrm{H}}$ during a transition window is
\begin{align}
    \frac{\partial h_{\mathrm{H}}(x,t)}{\partial t}
    &=
    \big(1-\eta_\ell(t)\big)
    \frac{\partial h_\ell^{-}(x,t)}{\partial t}
    +
    \eta_\ell(t)
    \frac{\partial h_{\ell+1}^{+}(x,t)}{\partial t}
    \notag\\
    &\quad
    +
    \dot{\eta}_\ell(t)
    \Big(
    h_{\ell+1}^{+}(x,t)-h_\ell^{-}(x,t)
    \Big),
    \qquad
    t\in\mathcal{T}_\ell.
    \label{eq:hH_time_derivative_transition}
\end{align}
The first two terms in \eqref{eq:hH_time_derivative_transition} capture the time variation of the pre- and post-switch SVM boundaries, while the last term captures transition-induced boundary motion. This last term quantifies how rapidly the barrier is moved from the pre-switch representation to the post-switch representation. On fixed-active-set intervals where no transition is active, $h_{\mathrm{H}}(x,t)=h_\ell(x,t)$. Hence, using \eqref{eq:h_ell_t_vs_form},
\begin{equation}
    \frac{\partial h_{\mathrm{H}}(x,t)}{\partial t}
    =
    \frac{\partial h_\ell(x,t)}{\partial t}
    =
    \Phi_\ell(x,t)v_s(t),
    \qquad
    t\in I_\ell\setminus \mathcal{T}_{\ell-1}.
    \label{eq:hH_time_derivative_fixed_interval}
\end{equation}
Thus, on fixed-active-set intervals, the boundary motion is governed by the SVM scheduling rate $v_s(t)$. The transition duration $T_\ell^{\mathrm{tr}}$ controls the rate at which the barrier moves between consecutive active-set representations. Since \(\dot{\eta}_\ell(t)=\eta'(s_\ell(t))/T_\ell^{\mathrm{tr}}\), let \(\bar{\eta}_1=\max_{s\in[0,1]}|\eta'(s)|\). For any compact set \(\mathcal{X}_c\subset\mathbb{R}^n\), define the transition mismatch bound \(\Delta_{\ell,c}=\sup_{(x,t)\in\mathcal{X}_c\times\mathcal{T}_\ell}\left|h_{\ell+1}^{+}(x,t)-h_\ell^{-}(x,t)\right|\). Then the transition-induced component of $\partial h_{\mathrm{H}}/\partial t$ satisfies
\begin{equation}
\left|
\dot{\eta}_\ell(t)
\Big(
h_{\ell+1}^{+}(x,t)-h_\ell^{-}(x,t)
\Big)
\right|
\leq
\frac{\bar{\eta}_1\Delta_{\ell,c}}{T_\ell^{\mathrm{tr}}},
\quad
(x,t)\in\mathcal{X}_c\times\mathcal{T}_\ell.
\label{eq:transition_speed_bound}
\end{equation}
The transition duration \(T_\ell^{\mathrm{tr}}\) is selected to keep the transition-induced boundary motion compatible with the degraded input authority. Smaller \(T_\ell^{\mathrm{tr}}\) yields faster switching but can make the CBF constraint more restrictive, while larger \(T_\ell^{\mathrm{tr}}\) reduces the derivative term at the cost of slower convergence to the post-switch representation. This feasibility effect is captured in the QP through \(\partial h_{\mathrm{H}}/\partial t\). Since \(h_{\mathrm{H}}(x,t)\) is time-varying, its total derivative along trajectories of \eqref{dyna} is
\begin{equation}
    \dot{h}_{\mathrm{H}}(x,t)
    =
    L_fh_{\mathrm{H}}(x,t)
    +
    L_gh_{\mathrm{H}}(x,t)u
    +
    \frac{\partial h_{\mathrm{H}}(x,t)}{\partial t},
    \label{eq:hH_dot_lie}
\end{equation}
where \(L_fh_{\mathrm{H}}(x,t)=\frac{\partial h_{\mathrm{H}}(x,t)}{\partial x}f(x)\) and \(L_gh_{\mathrm{H}}(x,t)=\frac{\partial h_{\mathrm{H}}(x,t)}{\partial x}g(x)\).
Let $\kappa:\mathbb{R}\to\mathbb{R}$ be a locally Lipschitz extended class-$\mathcal{K}$ function. The homotopy-smoothed SVM barrier $h_{\mathrm{H}}(x,t)$ is said to satisfy the time-varying CBF condition if, for each $x\in\widehat{\mathcal{S}}_{\mathrm{H}}(t)$, there exists a control input $u\in\mathcal{U}(t)$ such that
\begin{equation}
    L_fh_{\mathrm{H}}(x,t)
    +
    L_gh_{\mathrm{H}}(x,t)u
    +
    \frac{\partial h_{\mathrm{H}}(x,t)}{\partial t}
    \geq
    -\kappa\big(h_{\mathrm{H}}(x,t)\big).
    \label{eq:hH_tvcbf_condition}
\end{equation}
On fixed-active-set intervals where no transition is active, \eqref{eq:hH_tvcbf_condition} reduces to
\begin{equation}
    L_fh_\ell(x,t)
    +
    L_gh_\ell(x,t)u
    +
    \Phi_\ell(x,t)v_s(t)
    \geq
    -\kappa\big(h_\ell(x,t)\big).
    \label{eq:hH_tvcbf_condition_fixed_interval}
\end{equation}
During active-set transition windows, the derivative term is instead given by \eqref{eq:hH_time_derivative_transition}. The following lemma summarizes the regularity induced by the homotopy-smoothed construction.
\begin{lemma}
\label{lem:hH_regularity}
Suppose the active-set switching times \(\{\tau_\ell\}\) are locally finite, the transition windows \(\mathcal{T}_\ell=[\tau_{\ell+1},\tau_{\ell+1}+T_\ell^{\mathrm{tr}}]\) are nonoverlapping, and \(h_\ell\), \(h_\ell^{-}\), and \(h_{\ell+1}^{+}\) are \(C^1\) in \((x,t)\) on their respective domains. Let \(\eta\in C^1([0,1];[0,1])\) be nondecreasing with \(\eta(0)=0\), \(\eta(1)=1\), and \(\eta'(0)=\eta'(1)=0\). Let the transition extensions satisfy \(h_\ell^{-}(x,\tau_{\ell+1})=h_\ell(x,\tau_{\ell+1}^{-})\) and \(h_{\ell+1}^{+}(x,\tau_{\ell+1}+T_\ell^{\mathrm{tr}})=h_{\ell+1}(x,\tau_{\ell+1}+T_\ell^{\mathrm{tr}})\). Then \(h_{\mathrm{H}}\) is continuous in \((x,t)\), piecewise \(C^1\) in time, and continuously differentiable with respect to \(x\) on each smooth piece. Moreover, \(h_{\mathrm{H}}\) admits a time derivative for almost all \(t\), and for every absolutely continuous trajectory \(x(\cdot)\), the map \(t\mapsto h_{\mathrm{H}}(x(t),t)\) is absolutely continuous on compact time intervals and differentiable everywhere.
\end{lemma}
\begin{proof}
Let \([a,b]\subset\mathbb{R}_{\geq 0}\) be an arbitrary compact time interval. By local finiteness of the switching times and nonoverlap of the transition windows, \([a,b]\) can be decomposed into a finite union \([a,b]=\bigcup_{q=1}^{N_c}\mathcal{J}_q\) such that, on each \(\mathcal{J}_q\), either \(h_{\mathrm{H}}=h_\ell\) for some fixed-active-set representation or \(h_{\mathrm{H}}(x,t)=(1-\eta_\ell(t))h_\ell^{-}(x,t)+\eta_\ell(t)h_{\ell+1}^{+}(x,t)\). Since \(h_\ell\), \(h_\ell^{-}\), and \(h_{\ell+1}^{+}\) are \(C^1\) on their respective domains and \(\eta_\ell\in C^1\), \(h_{\mathrm{H}}\) is \(C^1\) on each smooth piece \(\mathcal{J}_q\). Continuity across transition endpoints follows from the endpoint matching conditions. At the beginning of the transition, \(\eta_\ell(\tau_{\ell+1})=0\), so \(h_{\mathrm{H}}(x,\tau_{\ell+1})=h_\ell^{-}(x,\tau_{\ell+1})=h_\ell(x,\tau_{\ell+1}^{-})\). At the end of the transition, \(\eta_\ell(\tau_{\ell+1}+T_\ell^{\mathrm{tr}})=1\), so \(h_{\mathrm{H}}(x,\tau_{\ell+1}+T_\ell^{\mathrm{tr}})=h_{\ell+1}^{+}(x,\tau_{\ell+1}+T_\ell^{\mathrm{tr}})=h_{\ell+1}(x,\tau_{\ell+1}+T_\ell^{\mathrm{tr}})\). Hence, \(h_{\mathrm{H}}\) is continuous across transition endpoints. Since only finitely many such endpoints occur on \([a,b]\), \(h_{\mathrm{H}}\) is continuous on \(\mathbb{R}^n\times[a,b]\), piecewise \(C^1\) in time, and admits a time derivative except possibly at finitely many endpoints.

Now let \(x(\cdot)\in AC([a,b];\mathbb{R}^n)\) and define \(s(t)=h_{\mathrm{H}}(x(t),t)\). On each \(\mathcal{J}_q\), \(h_{\mathrm{H}}\) is \(C^1\) and \(x(t)\) is absolutely continuous; hence \(s(t)\) is absolutely continuous on \(\mathcal{J}_q\), and for almost all \(t\in\mathcal{J}_q\), \(\dot{s}(t)=\frac{\partial h_{\mathrm{H}}}{\partial x}(x(t),t)\dot{x}(t)+\frac{\partial h_{\mathrm{H}}}{\partial t}(x(t),t)\). Since there are finitely many subintervals and \(s(t)\) is continuous at their endpoints, these absolutely continuous pieces concatenate to give \(s\in AC([a,b])\). Since \([a,b]\) was arbitrary, the result follows.
\end{proof}
\begin{theorem}
\label{thm:forward_invariance_tvcbf}
Let $h_{\mathrm{H}}(x,t)$ be the transition-smoothed SVM barrier defined by \eqref{eq:hH_homotopy_transition}--\eqref{eq:hH_fixed_interval}, and suppose the conditions of Lemma~\ref{lem:hH_regularity} hold. Let $\kappa:\mathbb{R}\to\mathbb{R}$ be a locally Lipschitz extended class-$\mathcal{K}$ function, and consider an absolutely continuous closed-loop trajectory $x(t)$ of \eqref{dyna} on $[t_0,T)$. If the applied input $u(t)\in\mathcal{U}(t)$ satisfies the time-varying CBF condition \eqref{eq:hH_tvcbf_condition} along the trajectory for almost all $t\in[t_0,T)$, and if $h_{\mathrm{H}}(x(t_0),t_0)\geq 0$, then $h_{\mathrm{H}}(x(t),t)\geq 0$ for all $t\in[t_0,T)$. Equivalently, the time-varying learned safe set $\widehat{\mathcal{S}}_{\mathrm{H}}(t)=\{x\in\mathbb{R}^n:h_{\mathrm{H}}(x,t)\geq 0\}$ is forward invariant along the closed-loop trajectory.
\end{theorem}
\begin{proof}
Define \(s(t)=h_{\mathrm{H}}(x(t),t)\). By Lemma~\ref{lem:hH_regularity}, \(s(t)\) is absolutely continuous on every compact subinterval of \([t_0,T)\) and differentiable almost everywhere. Moreover, for almost all \(t\in[t_0,T)\), the chain rule gives
\begin{equation}
    \dot{s}(t)
    =
    L_fh_{\mathrm{H}}(x(t),t)
    +
    L_gh_{\mathrm{H}}(x(t),t)u(t)
    +
    \partial_t h_{\mathrm{H}}(x(t),t).
\end{equation}
Using \eqref{eq:hH_tvcbf_condition}, we obtain \(\dot{s}(t)\geq -\kappa(s(t))\) for almost all \(t\in[t_0,T)\). Consider the scalar comparison system \(\dot{z}(t)=-\kappa(z(t))\), with \(z(t_0)=s(t_0)\). Since \(\kappa\) is locally Lipschitz, this scalar system admits a unique solution. Also, since \(\kappa(0)=0\), the origin is an equilibrium, and \(z(t_0)=s(t_0)\geq 0\) implies \(z(t)\geq 0\) for all \(t\in[t_0,T)\). By the scalar comparison principle for absolutely continuous functions, \(s(t)\geq z(t)\) for all \(t\in[t_0,T)\). Hence, \(s(t)\geq 0\) for all \(t\in[t_0,T)\). Since \(s(t)=h_{\mathrm{H}}(x(t),t)\), it follows that \(h_{\mathrm{H}}(x(t),t)\geq 0\) for all \(t\in[t_0,T)\). Thus, \(x(t)\in\widehat{\mathcal{S}}_{\mathrm{H}}(t)\) for all \(t\in[t_0,T)\), which proves forward invariance.
\end{proof}
Theorem~\ref{thm:forward_invariance_tvcbf} shows that enforcing the time-varying CBF condition for \(h_{\mathrm{H}}(x,t)\) guarantees safety. On fixed-active-set intervals, this reduces to \eqref{eq:hH_tvcbf_condition_fixed_interval}, while across active-set transitions, the homotopy-smoothed barrier keeps the condition well defined for the forward-invariance argument.
\begin{figure}[htb]
    \begin{center}
        \includegraphics[width=1\linewidth]{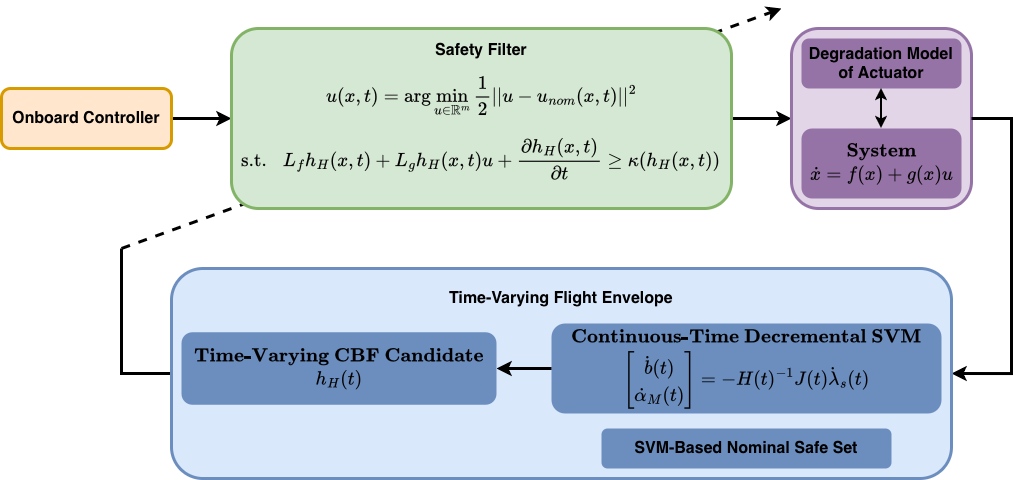}
    \end{center}
    \caption{Degradation-aware SVM-CBF-QP framework for adaptive safe-set learning and safe control under reduced control authority.}
    \label{fig:intro2}
\end{figure}
\subsection{Degradation-Aware SVM-CBF-QP Safety Filter}
\label{subsec:qp_safety_filter}

Fig.~\ref{fig:intro2} summarizes the proposed degradation-aware SVM-CBF-QP framework for adaptive safe-set learning and safe control under reduced control authority. To enforce the homotopy-smoothed time-varying SVM-CBF condition while preserving nominal control performance, we construct a safety filter that minimally modifies a nominal input. Let \(u_{\mathrm{nom}}(x,t)\) denote a nominal control input designed. The safety-filtered input is computed by solving
\begin{equation}
\begin{aligned}
    u^\star(x,t)
    =
    \arg\min_{u\in\mathbb{R}^m}\quad
    & \frac{1}{2}\|u-u_{\mathrm{nom}}(x,t)\|^2 \\
    \mathrm{s.t.}\quad
    &
    L_fh_{\mathrm H}(x,t)
    +
    L_gh_{\mathrm H}(x,t)u
    \\
    &+
    \partial_t h_{\mathrm H}(x,t)
    \geq
    -\kappa\big(h_{\mathrm H}(x,t)\big), \\
    &
    A_u u\leq \lambda(t)b_u .
\end{aligned}
\label{eq:qp_safety_filter}
\end{equation}
The first constraint in \eqref{eq:qp_safety_filter} enforces the time-varying CBF condition for the homotopy-smoothed SVM barrier, while the second enforces the degraded input constraint \(\mathcal{U}(t)=\{u\in\mathbb{R}^m:A_u u\leq \lambda(t)b_u\}\). The strictly convex objective makes the filter minimally invasive, and the affine constraints make \eqref{eq:qp_safety_filter} a convex QP with a unique optimizer whenever its feasible set is nonempty. On fixed-active-set intervals, \(h_{\mathrm H}=h_\ell\) and \(\partial_t h_{\mathrm H}=\Phi_\ell v_s\), so the CBF constraint becomes \(L_fh_\ell(x,t)+L_gh_\ell(x,t)u+\Phi_\ell(x,t)v_s(t)\geq-\kappa(h_\ell(x,t))\). During transition windows, \(\partial_t h_{\mathrm H}\) is computed from \eqref{eq:hH_time_derivative_transition}, allowing the QP to account for both degradation-driven and transition-induced boundary motion. For notational compactness, define \(a_{\mathrm H}(x,t)=L_gh_{\mathrm H}(x,t)\in\mathbb{R}^{1\times m}\) and \(r_{\mathrm H}(x,t)=-\kappa(h_{\mathrm H}(x,t))-L_fh_{\mathrm H}(x,t)-\partial_t h_{\mathrm H}(x,t)\). Then the barrier constraint in \eqref{eq:qp_safety_filter} is equivalent to \(a_{\mathrm H}(x,t)u\geq r_{\mathrm H}(x,t)\). The following lemma characterizes pointwise feasibility of the CBF constraint under the degraded polytopic input constraint.

\begin{lemma}
\label{lem:pointwise_feasibility_H}
Fix \((x,t)\), and let \(a_{\mathrm H}=a_{\mathrm H}(x,t)\) and \(r_{\mathrm H}=r_{\mathrm H}(x,t)\). Let \(\mathcal{U}(t)=\{u\in\mathbb{R}^m:A_u u\leq \lambda(t)b_u\}\). Then there exists \(u\in\mathcal U(t)\) satisfying the SVM-CBF constraint \(L_fh_{\mathrm H}+L_gh_{\mathrm H}u+\partial_t h_{\mathrm H}\geq-\kappa(h_{\mathrm H})\) if and only if \(r_{\mathrm H}\leq\sigma_{\mathcal U(t)}(a_{\mathrm H})\), where \(\sigma_{\mathcal U(t)}(a_{\mathrm H})=\max_{u\in\mathbb{R}^m}\{a_{\mathrm H}u:A_u u\leq \lambda(t)b_u\}\) is the support function of the degraded input polytope in the direction \(a_{\mathrm H}\).
\end{lemma}
\begin{proof}
At the fixed point \((x,t)\), the SVM-CBF constraint is equivalent to \(a_{\mathrm H}u\geq r_{\mathrm H}\). Therefore, a feasible input exists if and only if the maximum achievable value of \(a_{\mathrm H}u\) over the degraded input set \(\mathcal U(t)\) is at least \(r_{\mathrm H}\). Since \(\max_{u\in\mathcal U(t)} a_{\mathrm H}u=\sigma_{\mathcal U(t)}(a_{\mathrm H})\), the feasibility condition is \(r_{\mathrm H}\leq\sigma_{\mathcal U(t)}(a_{\mathrm H})\). If \(a_{\mathrm H}=0\), then \(\sigma_{\mathcal U(t)}(a_{\mathrm H})=0\), and the condition reduces to \(r_{\mathrm H}\leq 0\), which means that the CBF inequality must be satisfied independently of the input. This completes the proof.
\end{proof}
Lemma~\ref{lem:pointwise_feasibility_H} shows that feasibility depends on \(r_{\mathrm H}\), \(\mathcal U(t)\), and \(a_{\mathrm H}=L_gh_{\mathrm H}\). Since decreasing \(\lambda(t)\) contracts \(\mathcal U(t)\) and reduces corrective authority, the learned safe set and homotopy-smoothed barrier must remain compatible with the remaining control authority.

\begin{theorem}
\label{thm:recursive_feasibility}
Let \(h_{\mathrm H}\) satisfy the regularity conditions in Lemma~\ref{lem:hH_regularity}. Suppose that, for all \(x\in\widehat{\mathcal S}_{\mathrm H}(t)\), the pointwise feasibility condition
\[
    -\kappa\big(h_{\mathrm H}(x,t)\big)
    -
    L_fh_{\mathrm H}(x,t)
    -
    \partial_t h_{\mathrm H}(x,t)
    \leq
    \sigma_{\mathcal U(t)}
    \big(
    L_gh_{\mathrm H}(x,t)
    \big)
\]
holds for almost all \(t\in[t_0,T)\). If \(x(t_0)\in\widehat{\mathcal S}_{\mathrm H}(t_0)\), then the safety filter \eqref{eq:qp_safety_filter} is feasible for almost all \(t\in[t_0,T)\) along the closed-loop trajectory, and \(x(t)\in\widehat{\mathcal S}_{\mathrm H}(t)\) for all \(t\in[t_0,T)\).
\end{theorem}
\begin{proof}
Fix any \(t\in[t_0,T)\) for which the assumed pointwise feasibility condition holds, and suppose \(x(t)\in\widehat{\mathcal S}_{\mathrm H}(t)\). By Lemma~\ref{lem:pointwise_feasibility_H}, there exists at least one input \(u\in\mathcal U(t)\) satisfying the SVM-CBF constraint. Hence, the constraint set of the safety filter \eqref{eq:qp_safety_filter} is nonempty, and the QP is feasible at time \(t\). Since the optimizer \(u^\star(x,t)\) satisfies the SVM-CBF constraint by construction, the closed-loop trajectory satisfies the time-varying CBF condition in \eqref{eq:hH_tvcbf_condition} for almost all \(t\in[t_0,T)\). Therefore, by Theorem~\ref{thm:forward_invariance_tvcbf}, the initial condition \(x(t_0)\in\widehat{\mathcal S}_{\mathrm H}(t_0)\) implies
\[
    x(t)\in\widehat{\mathcal S}_{\mathrm H}(t),
    \qquad
    \forall t\in[t_0,T).
\]
Since the trajectory remains in \(\widehat{\mathcal S}_{\mathrm H}(t)\), and the assumed pointwise feasibility condition holds for all states in \(\widehat{\mathcal S}_{\mathrm H}(t)\), the QP remains feasible along the trajectory for almost all \(t\in[t_0,T)\). This proves recursive feasibility.
\end{proof}
\begin{algorithm}[t]
\caption{Degradation-Aware Time-Varying SVM-CBF Safety Filter}
\label{alg:tvcbf}
\begin{algorithmic}[1]
\REQUIRE Training data $\{(x_i,y_i)\}_{i=1}^{N}$; nominal controller $u_{\mathrm{nom}}$; degradation profile $\lambda(t)$; rate $k_c$; window $T^{\mathrm{tr}}$
\ENSURE Safe control input $u^{\star}(t)$
\STATE \textbf{Offline initialization:}
\STATE Train RBF-SVM by solving \eqref{eq:svm_dual_problem} to obtain $\{\alpha_i\},b$ and $h_0$ in \eqref{eq:nominal_decision_function}
\STATE Partition samples into $\mathcal{M},\mathcal{E},\mathcal{R}$ via \eqref{eq:nominal_margin_set}--\eqref{eq:nominal_reserve_set}
\STATE Set $c\gets\varnothing$, $\mathrm{blend}\gets\mathrm{false}$
\STATE \textbf{Online safety filtering:}
\FOR{each sampling instant $t$}
    \STATE Update $\lambda(t)$, $\lambda_s(t)$, $v_s(t)$ via \eqref{eq:lambdas_dynamics}, \eqref{eq:vs_buffered_rate}
    \IF{$v_s(t)>0$}
        \IF{$c=\varnothing$ \textbf{or} $\alpha_c\le 0$}
            \STATE Select next safe support vector $c$ by weighted $(\alpha,q)$ score; exclude $c$ from $\mathcal{M}$
        \ENDIF
        \STATE Build $H(t),J(t)$ from \eqref{eq:H_explicit}, \eqref{eq:H_J_defs} and solve $\dot\theta=H^{-1}J\,v_s$ \eqref{eq:boundary_parameter_dynamics}
    \ENDIF
    \STATE Evaluate $h_\ell$ \eqref{eq:local_tv_svm_decision_function}, $\nabla_x h_\ell$ \eqref{eq:grad_h_ell}, $\partial_t h_\ell=\Phi_\ell v_s$ \eqref{eq:h_ell_t_vs_form}
    \IF{$\mathrm{blend}$ active}
        \STATE Form $h_{\mathrm H}$ \eqref{eq:hH_homotopy_transition} and $\partial_t h_{\mathrm H}$ \eqref{eq:hH_time_derivative_transition} by blending pre-/post-switch boundaries
    \ELSE
        \STATE $h_{\mathrm H}\gets h_\ell$, $\partial_t h_{\mathrm H}\gets\partial_t h_\ell$ \eqref{eq:hH_fixed_interval}
    \ENDIF
    \STATE Compute $L_f h_{\mathrm H}, L_g h_{\mathrm H}$ and form the CBF constraint \eqref{eq:hH_tvcbf_condition}
    \STATE Solve QP \eqref{eq:qp_safety_filter} for $u^{\star}$; apply to \eqref{dyna}
    \STATE Reduce $\alpha_c\!\gets\!\max(\alpha_c-k_c v_s\,dt,0)$; update $\theta\!\gets\!\theta+\dot\theta\,dt$ \eqref{eq:boundary_parameter_dynamics}
    \STATE Reclassify $\mathcal{M},\mathcal{E},\mathcal{R}$ via \eqref{eq:nominal_margin_set}--\eqref{eq:nominal_reserve_set}
    \IF{$\mathcal{M}$ changed \textbf{and not} $\mathrm{blend}$}
        \STATE $\mathrm{blend}\gets\mathrm{true}$; initialize transition window \eqref{eq:hH_homotopy_transition}
    \ENDIF
\ENDFOR
\end{algorithmic}
\end{algorithm}
\section{Simulation Results}
The proposed degradation-aware SVM-CBF-QP safety-filter framework is evaluated 
using the linearized short-period longitudinal dynamics of a VTOL aircraft 
operating under cruise conditions. The corresponding state-space model is 
given by
\begin{equation}
    \dot{x} = Ax + B u,
    \qquad
    y = Cx,
\end{equation}
where the state vector is $x = [\alpha \quad q]^T$, with $\alpha$ denoting the angle of attack and $q$ denoting the pitch rate. Also, $u$ represents the elevator deflection. The system matrices are
\begin{equation}
    A =
    \begin{bmatrix}
        -0.394 & 0.993 \\
        -1.619 & -0.395
    \end{bmatrix},
    \;\;
    B =
    \begin{bmatrix}
        -0.021 \\ -1.214
    \end{bmatrix},
    \;\;
    C =
    \begin{bmatrix}
        1 & 0
    \end{bmatrix}.
\end{equation}
Under nominal operating conditions, the control input is constrained by \(|u|\leq 0.3~\text{rad}\), and the state constraints are \(|\alpha|\leq 0.25~\text{rad}\) and \(|q|\leq 0.25~\text{rad/s}\). A baseline Linear Quadratic Integral (LQI) \cite{Lavretsky} controller is employed for angle-of-attack tracking. Integral action is introduced by augmenting the system with the integral state $\xi$, defined by
\begin{equation}
    \dot{\xi} = r - y,
\end{equation}
where $r$ is the reference angle of attack. The augmented system is
\begin{equation}
    A_a =
    \begin{bmatrix}
        A & \mathbf{0} \\ 
        -C & 0
    \end{bmatrix},
    \qquad
    B_a =
    \begin{bmatrix}
        B \\ 0
    \end{bmatrix}.
\end{equation}
The LQI controller is designed using the weighting matrices $Q=\mathrm{diag}(10,1,50)$ and $R=1$. The optimal feedback gain is obtained by solving the continuous-time algebraic Riccati equation for the augmented system, resulting in the control law
\begin{equation}
    u = -K_x x - K_i \xi,
\end{equation}
where $K = [\,K_x \;\; K_i\,]$ is the optimal LQI gain matrix. The reference angle of attack is defined as
\begin{equation}
    \alpha_{\mathrm{ref}}(t)=
    \begin{cases}
        0, & t < 3~\text{s},\\[2mm]
        0.25\sin\!\left(\dfrac{2\pi(t-3)}{6}\right), & t \geq 3~\text{s},
    \end{cases}
\end{equation}
which introduces a sinusoidal tracking command after 3 seconds.
\subsection{Initial Safe Set Learning and Simulation Setup}
The nominal safe envelope is learned from $N=225$ randomnly generated labeled samples on a uniform grid, with angle of attack $\alpha$ (rad) on the $x$-axis and pitch rate $q$ (rad/s) on the $y$-axis, both ranging from $-0.4$ to $0.4$. Samples are labeled safe if $|\alpha|\leq 0.25$~rad and $|q|\leq 0.25$~rad/s, and unsafe otherwise. A soft-margin RBF-SVM with $\gamma=30$ and $C=1$ is trained on this dataset, and its decision function $h_0(x)$ in \eqref{eq:nominal_decision_function} defines the initial barrier candidate. The admissible input set is the box $\mathcal{U}(t)=\{u: A_u u\le\lambda(t)b_u\}$, with $A_u=[1;-1]$ and $b_u=[u_{\max};u_{\max}]$. The scheduling signal is generated with $\varepsilon_s=0.05$ and $\lambda_{\min}=0.55$, ensuring that $\lambda_s(t)\le\lambda(t)$. The decremental index $c$ is selected sequentially throughout the simulation. At each cycle, $c$ is manually selected from the safe-labeled support vectors according to the weighted score $\alpha_c^2+60q_c^2$, with the support vector having the largest score chosen for removal. This weighting encourages the removal of the outer support vectors with greater influence in the $q$-direction, reflecting the stronger effect of elevator deflection degradation on $q$. Its coefficient is then reduced at a constant rate $k_c=130$, tuned specifically for this case to obtain the desired shrinking behavior, until $\alpha_c=0$ after which the next support vector is selected. Although the active set changes after each selection, the barrier function, $h(x,t)$ remains continuous. However, its time derivative $\partial h/\partial t$ may be discontinuous during switching due to the change in the active coefficient, consistent with the piecewise-$C^1$ regularity established in Lemma~\ref{lem:hH_regularity}. The homotopy transition \eqref{eq:hH_homotopy_transition} is therefore applied only when migration of the margin set changes the boundary representation. In this case, the homotopy function $\eta(s)=s^2(3-2s)$.
\subsection{Case 1: Degradation-Aware Safety Filtering}
For this case, the degradation signal decreases monotonically from $\lambda_s=1$ to $\lambda_{\min}$ over $t \in [5,25]$~s. Three configurations are compared under identical degraded input limits: (i) no safety filter, (ii) a time-invariant filter enforcing the nominal barrier $h_0$, and (iii) the proposed time-varying filter enforcing $h_H(x,t)$ with the continuous-time decremental update in \ref{subsec:ct_decremental_update}.
\begin{figure}[htb]
    \begin{center}
        \includegraphics[width=0.7\linewidth]{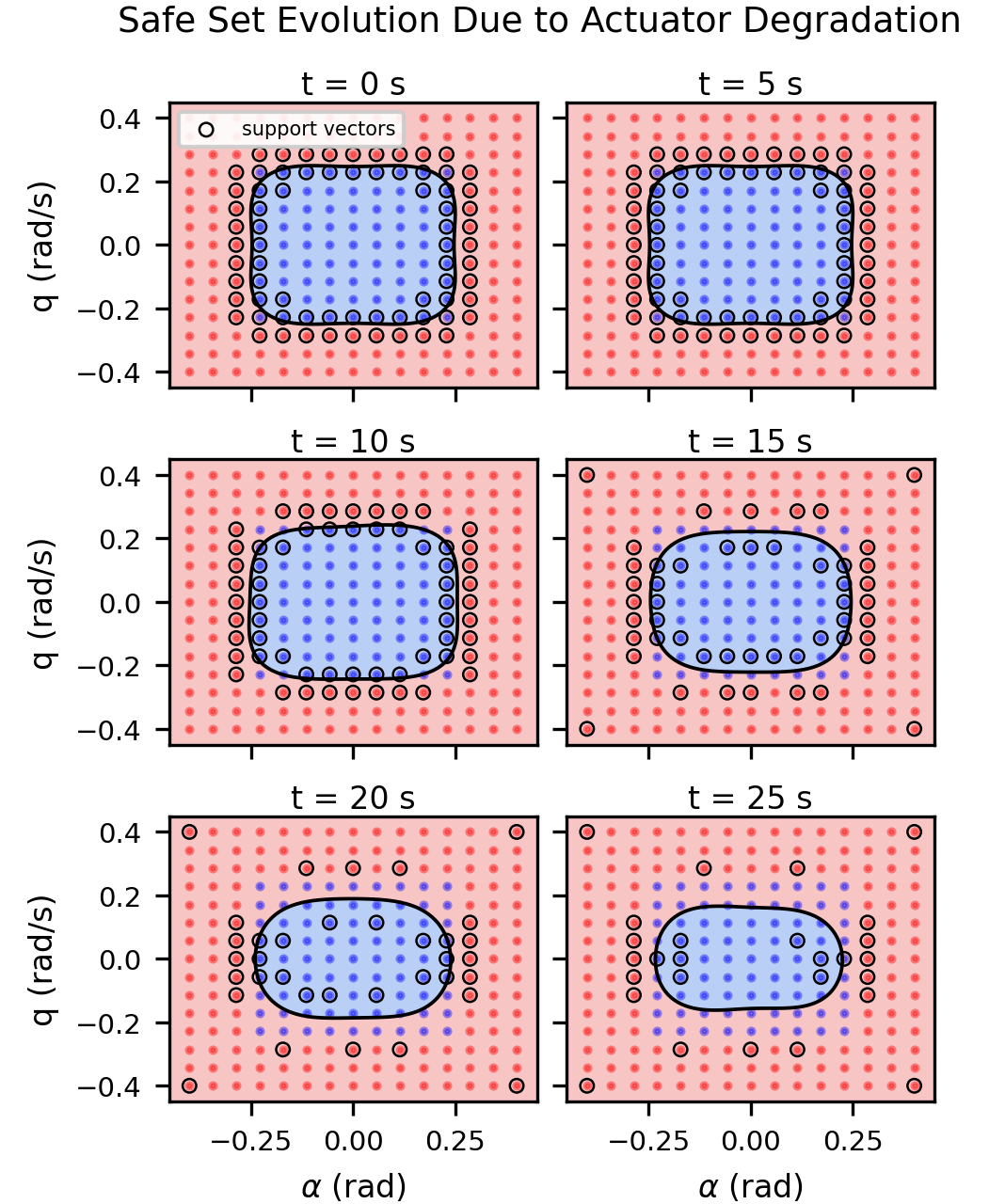}
    \end{center}
    \caption{Evolution of the learned safe envelope under actuator degradation. Circles denote support vectors, and the decremental update removes outer safe support vectors, contracting the envelope.}
    \label{fig:res1}
\end{figure}
Fig.~\ref{fig:res1} shows the evolution of the learned safe set at $t\in\{0,5,10,15,20,25\}$~s. The decremental update progressively contracts the envelope as support-vector influence is removed. Fig.~\ref{fig:res1b} shows the closed-loop trajectory with the time-varying filter enabled, overlaid on the learned safe set at $t\in\{0,5,10,15,20,25\}$~s. During nominal operation, the sinusoidal reference drives the state near the boundary of the initial safe set. As the safe set contracts due to degradation, the trajectory becomes more restricted. Although the earlier trajectory (dashed) lies outside the current safe set, the current state (dot) remains inside at every snapshot. From $t=15$ to $25$~s, the trajectory follows the contracting boundary, indicating that the filter limits the tracking response to satisfy the shrinking safety constraints. This behavior is consistent with the forward-invariance guarantee in Theorem~\ref{thm:forward_invariance_tvcbf}.
\begin{figure}[htb]
    \begin{center}
        \includegraphics[width=0.7\linewidth]{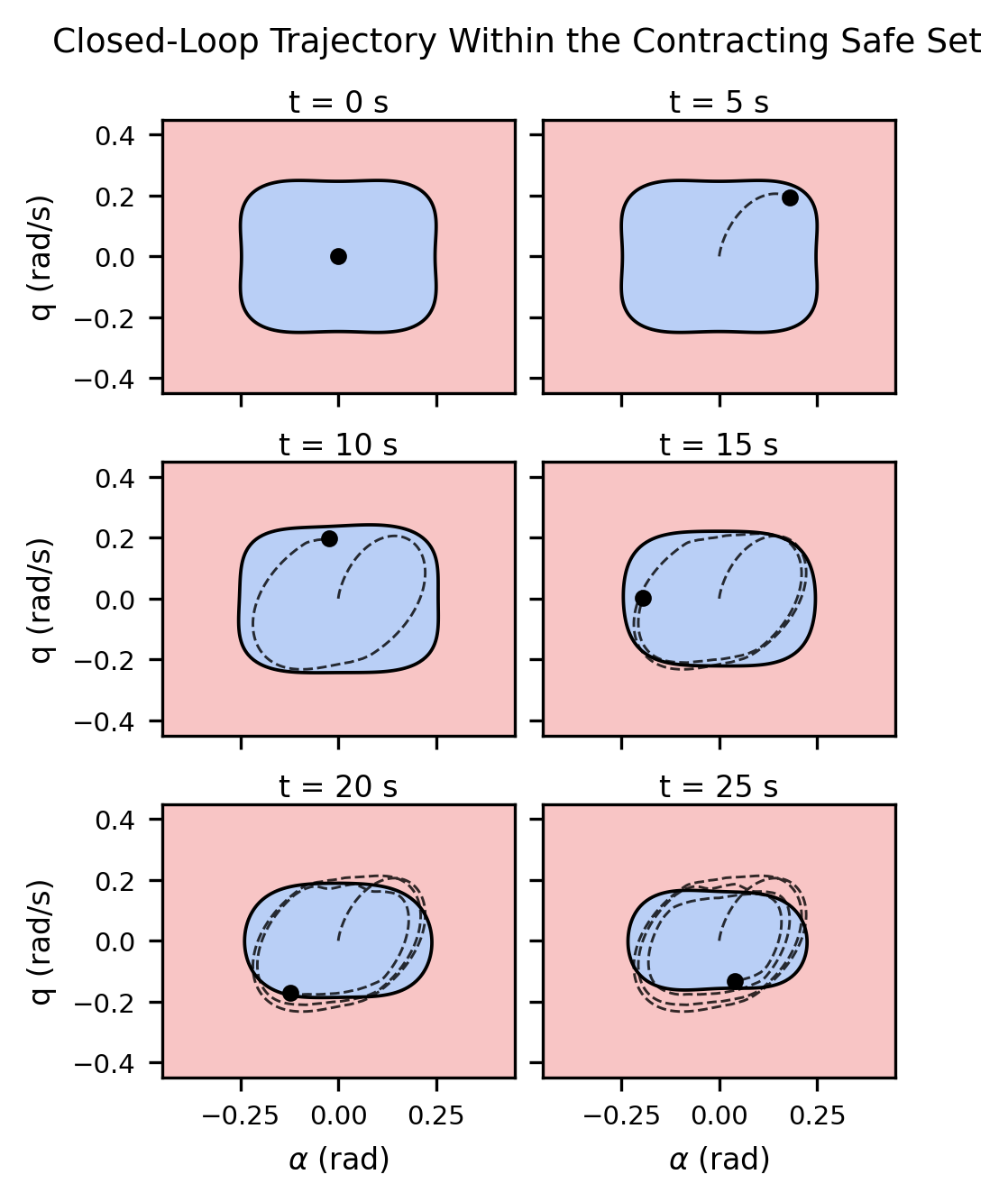}
    \end{center}
    \caption{Closed-loop trajectory under the proposed time-varying filter overlaid on the shrinking safe set at selected time. The dot marks the current state, and the dashed line denotes the past trajectory.}
    \label{fig:res1b}
\end{figure}
Figs.~\ref{fig:res2} and~\ref{fig:res3} compare the three configurations. Without a safety filter, the sinusoidal tracking command drives the state outside the contracted safe envelope ($\min_t h < 0$). The time-invariant filter enforces only the nominal barrier and therefore also violates the degraded envelope as $\lambda(t)$ decreases. In contrast, the proposed filter maintains $h_H(x(t),t)\ge 0$ for all $t$, and the QP remains feasible throughout the simulation. As shown in Fig.~\ref{fig:res3}, the intervention signal $u^\star-u_{\mathrm{nom}}$ is zero whenever the nominal command is safe, confirming the minimally invasive behavior of the safety filter in Section~\ref{subsec:qp_safety_filter}.
\begin{figure}[h]
    \centering
    \includegraphics[width=0.65\linewidth]{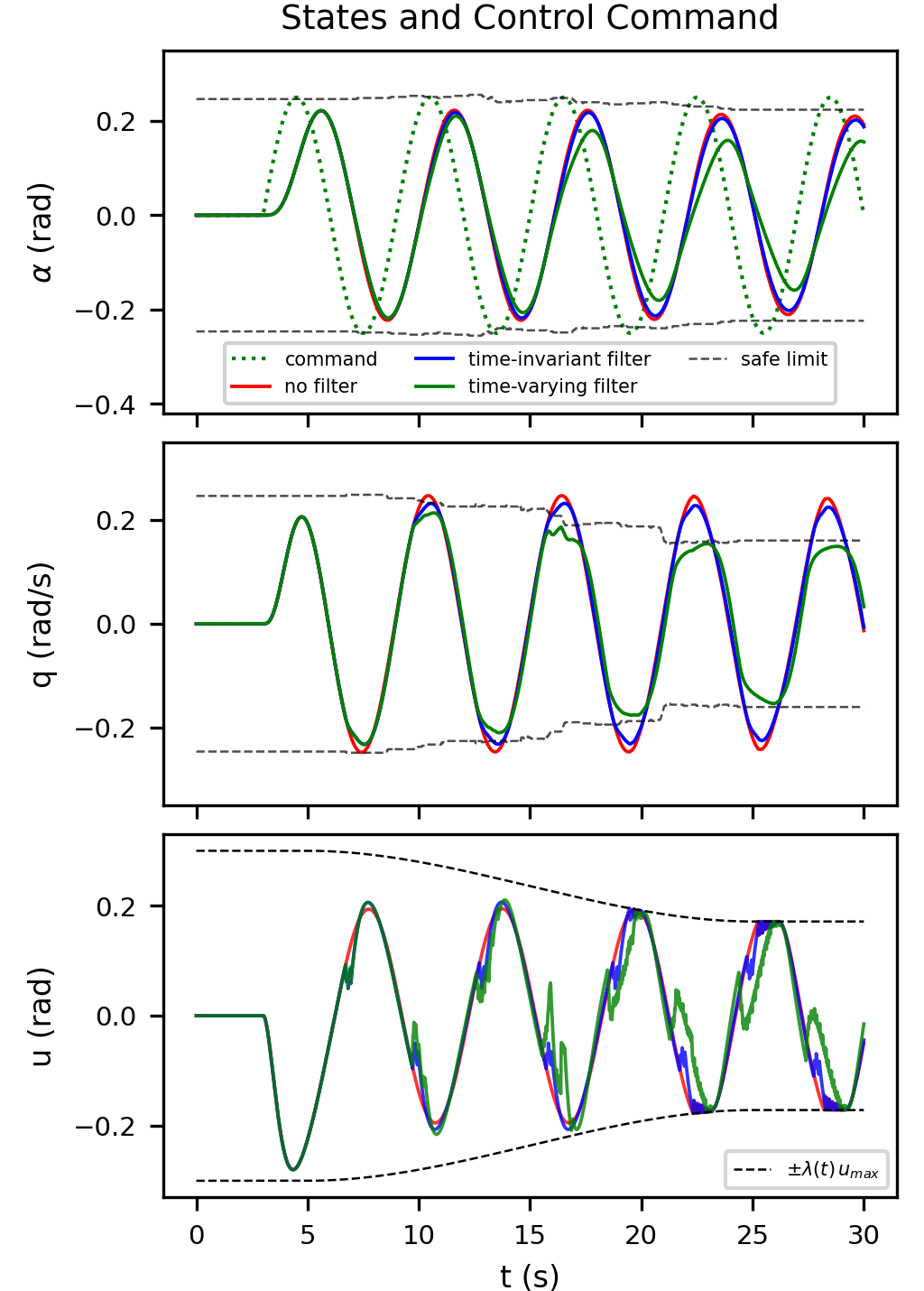}
    \caption{State trajectories and elevator command under smooth actuator degradation. Dashed lines indicate the contracting safe limits and degraded input bound.}
    \label{fig:res2}
\end{figure}
\begin{figure}[h]
    \centering
    \includegraphics[width=0.65\linewidth]{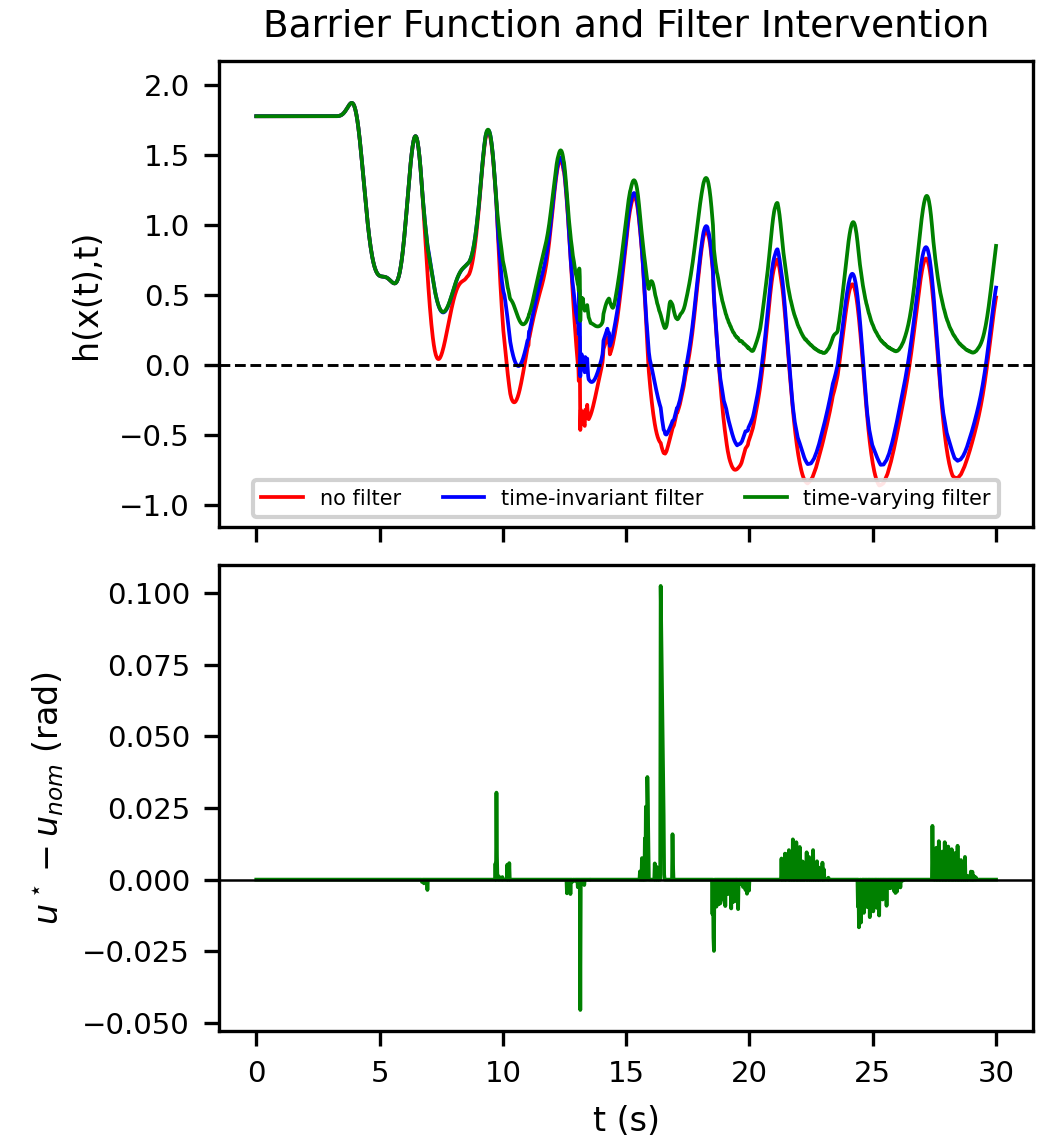}
    \caption{Barrier function and filter intervention. The proposed filter maintains $h_H(x(t),t)\ge0$ while intervening only when necessary.}
    \label{fig:res3}
\end{figure} 
\subsection{Case 2: Homotopy-Based Barrier Transition Under Abrupt Safe-Set Contraction}
To evaluate the homotopy mechanism under a worst-case representation switch, the initial and final learned sets from Case~1 are stored. At $t=5$~s, the barrier representation is switched instantaneously from the nominal set to the fully contracted set, while $\lambda$ simultaneously steps from $1$ to $0.571$. Two configurations are compared: an instantaneous barrier swap (without homotopy) and the homotopy-smoothed transition \eqref{eq:hH_homotopy_transition} with $T^{\mathrm{tr}}=1.0$~s. Figs.~\ref{fig:case2_states} and~\ref{fig:case2_barrier} summarize the results. Under the instantaneous swap, the enforced barrier jumps by $1.30$ at the switching instant and reaches $h=-0.63$. The state is therefore instantaneously outside the new safe set, violating the initialization hypothesis of Theorem~\ref{thm:forward_invariance_tvcbf}. Consequently, the swapped barrier is not a valid CBF at the switching instant. The enforced state limits also jump discontinuously, and the commanded input exhibits a discontinuity of $0.136$~rad. In contrast, the homotopy maintains a continuous barrier with $|\Delta h|\le10^{-3}$ across the switch and satisfies $h_H\ge0$ throughout. The transition term $\dot\eta_\ell(h^+-h^-)$ in \eqref{eq:hH_time_derivative_transition} tightens the CBF constraint during the blend, commanding an early corrective input
within the degraded limits. Although this produces a brief input peak, the control command remains admissible, and the maximum instantaneous input jump is reduced from $0.136$ to $0.055$~rad.
\begin{figure}[h]
    \centering
    \includegraphics[width=0.7\linewidth]{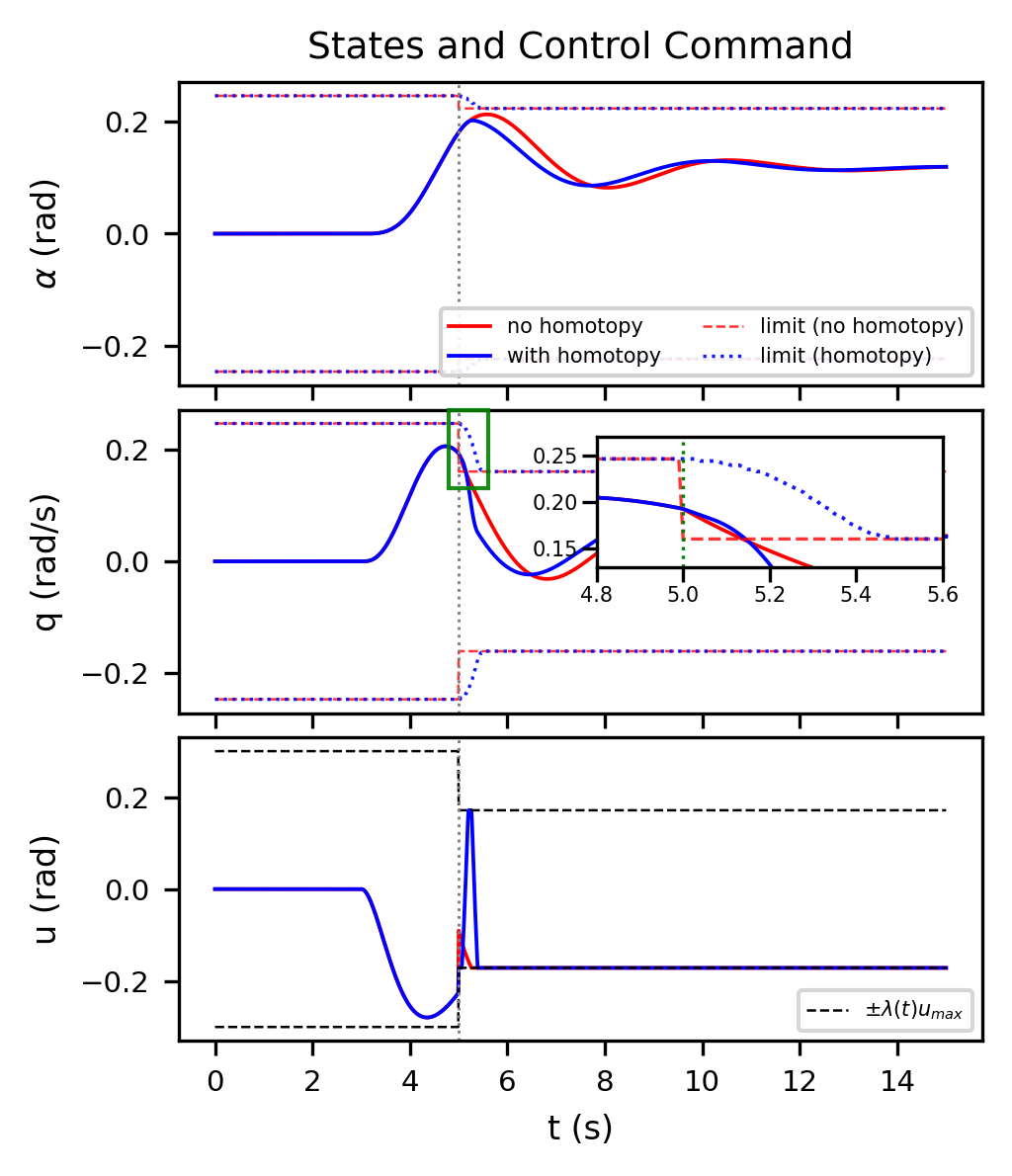}
    \caption{State trajectories and elevator command under an instantaneous barrier switch at $t=5$~s, with and without the homotopy transition. Dashed and dotted lines denote the enforced state limits.}
    \label{fig:case2_states}
\end{figure}
\begin{figure}[h]
    \centering
    \includegraphics[width=0.7\linewidth]{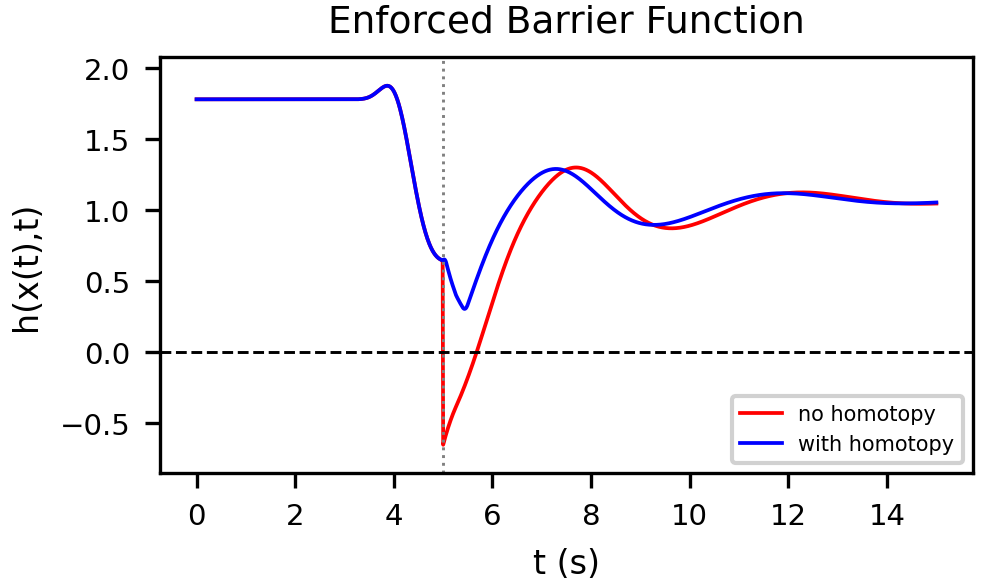}
    \caption{Enforced barrier function evaluation. The instantaneous swap produces a jump to $h<0$, while the homotopy keeps the barrier continuous and nonnegative.}
    \label{fig:case2_barrier}
\end{figure}
Table~\ref{tab:res1} shows the effect of the transition duration on QP feasibility. For small $T^{\mathrm{tr}}$, the transition term exceeds the available control authority under $\lambda u_{\max}$, causing the QP to become infeasible. Increasing $T^{\mathrm{tr}}$ restores feasibility. These results validate the bound in \eqref{eq:transition_speed_bound} and illustrate the trade-off discussed in Section~\ref{sec:safety_filter_design}.
\begin{table}[h]
\centering
\caption{Effect of transition duration on QP feasibility (Case 2).}
\label{tab:res1}
\begin{tabular}{ccc}
\hline
$T^{\mathrm{tr}}$ (s) & Infeasible steps & $\max_t|\Delta u|$ (rad) \\
\hline
0.05 & 4  & 0.3429 \\
0.20 & 14 & 0.1581 \\
0.50 & 7  & 0.055 \\
1.00 & 0  & 0.055 \\
2.00 & 0  & 0.055 \\
\hline
\end{tabular}
\end{table}
In summary, the simulations validate the proposed framework. Case 1 shows that the continuous-time decremental update shrinks the learned safe envelope as the actuator degrades, and that the resulting time-varying SVM-CBF-QP filter maintains $h_H(x(t),t)\ge0$ with a feasible QP throughout the simulation. On the other hand, the unfiltered and time-invariant configurations violate the degraded envelope. Case 2 demonstrates that an instantaneous barrier swap fails to satisfy the initialization hypothesis of Theorem~\ref{thm:forward_invariance_tvcbf}, whereas the homotopy transition preserves barrier continuity and reduces the instantaneous input jump. Finally, the transition-duration study confirms the feasibility trade-off predicted by \eqref{eq:transition_speed_bound}, showing that sufficiently large $T^{\mathrm{tr}}$ prevents infeasibility during active-set transitions.

\section{Conclusion and Future Works}
This work developed a degradation-aware, time-varying SVM-CBF-QP safety-filter framework. The nominal safe operating envelope was learned from operational data using an RBF-kernel SVM, and its decision function was used as the initial control barrier function. As the system degrades, the barrier is updated in real time with an online continuous-time decremental SVM update law, in which the support vector coefficients change according to a degradation scheduling signal while satisfying the active KKT conditions. This ensured that the barrier remained continuously differentiable between active-set transitions. To address the regularization issue introduced by the active-set changes, the paper introduces a homotopy-based smoothing method that yields a continuous-time-varying CBF. The learned CBF was then enforced with a QP-based safety filter subject to the degraded input constraints. The analysis established forward invariance of the homotopy-smoothed safe set together with recursive feasibility of the safety filter. Simulation studies using a VTOL short-period model showed that the proposed approach maintained safety even with a loss of elevator authority due to degradation. In comparison, the unfiltered system and the QP safety filter enforcing the initial safe set exceeded the degraded safe envelope. Also, the simulations showed that the homotopy transition preserved barrier continuity and removed input discontinuities during rapid contractions of the safe set, in agreement with the theoretical feasibility analysis. Future work will focus on extending the proposed framework to more advanced applications. The method will first be evaluated on full nonlinear, higher-dimensional vehicle models to assess its scalability and performance. It will then be validated through hardware-in-the-loop and flight experiments. Finally, the framework will be integrated with a real-time health-monitoring system so that the degradation scheduling signal is generated from estimated actuator health instead of a predefined degradation profile.

\section*{Acknowledgements}

This work has been supported by the NASA University Leadership Initiative under grant No. 80NSSC25M7104.

\bibliographystyle{IEEEtran}
\bibliography{shawon}

\end{document}